\PassOptionsToPackage{table}{xcolor}
\documentclass{article}
\usepackage{iclr2026_conference,times}

\usepackage{amsmath,amssymb,amsthm,mathtools}
\usepackage{aliascnt}
\usepackage{booktabs}
\usepackage{multirow}
\usepackage{graphicx}
\usepackage{xcolor}
\usepackage{enumitem}
\usepackage{algorithm,algpseudocode}
\usepackage{placeins}
\usepackage{microtype}
\usepackage{tikz}
\usetikzlibrary{arrows.meta,positioning,fit,backgrounds,calc,shapes.geometric}
\usepackage{hyperref}
\usepackage{url}
\usepackage[capitalise,noabbrev]{cleveref}

\definecolor{fovblue}{HTML}{1B4F72}
\definecolor{fovred}{HTML}{C0392B}

\newcommand{\reals}{\mathbb{R}}
\newcommand{\integers}{\mathbb{Z}}
\newcommand{\Exp}{\mathbb{E}}
\newcommand{\Prb}{\Pr}
\DeclareMathOperator{\supp}{supp}
\DeclareMathOperator{\rank}{rank}
\DeclareMathOperator{\diam}{diam}
\DeclareMathOperator{\per}{per}
\DeclareMathOperator{\dist}{dist}
\DeclareMathOperator{\nnz}{nnz}
\DeclareMathOperator{\intr}{int}

\newcommand{\vars}{[n]}                        % variable index set
\newcommand{\cons}{[m]}                        % constraint index set
\newcommand{\inc}{x^{\star}}                   % incumbent
\newcommand{\opt}{v^{\star}}                   % reference objective value
\newcommand{\sub}[1]{\mathrm{SUB}(#1,\inc)}    % restricted subproblem
\newcommand{\feas}{\mathcal{F}}                % feasible set
\newcommand{\gain}{\Delta}                     % improvement of a destroy set

\newcommand{\cut}[1]{\mathrm{cut}(#1)}
\newcommand{\cutE}[1]{\mathrm{cut}_{E}(#1)}
\newcommand{\inE}[1]{\mathrm{in}_{E}(#1)}
\newcommand{\Srand}{S_{\mathrm{rand}}}
\newcommand{\seff}{s_{\mathrm{eff}}}
\newcommand{\thetaE}{\theta_{E}}

\newcommand{\lay}{\varphi}                     % variable layout
\newcommand{\ents}{\mathcal{E}}                % entity set
\DeclareMathOperator{\ent}{ent}
\DeclareMathOperator{\cell}{cell}
\newcommand{\Vpq}{V_{pq}}
\newcommand{\canvas}{\mathcal{C}}
\newcommand{\Psel}{P_{\mathrm{sel}}}
\newcommand{\region}{R}
\newcommand{\SR}{S_{\region}}
\newcommand{\Rlp}{R_{\mathrm{lp}}}             % dual refresh period, in rounds
\newcommand{\regfam}{\mathcal{R}_{C,B}}        % admissible region family (A4)

\newcommand{\cdelta}{c_{\delta}}
\newcommand{\cg}{c_{g}}
\newcommand{\Bpar}{B_{\partial}}               % boundary components
\newcommand{\Bcc}{B_{\mathrm{cc}}}             % connected components
\newcommand{\Lampar}{\Lambda_{\partial}}       % mean density on the boundary
\newcommand{\Lammax}{\Lambda_{\max}}           % global maximum density
\newcommand{\rhomax}{\rho_{\max}}
\newcommand{\Sstar}{S^{\star}}
\newcommand{\Rstar}{R^{\star}}
\newcommand{\Bstar}{B^{\star}}
\newcommand{\gstar}{g^{\star}}
\newcommand{\nstar}{n^{\star}}
\newcommand{\rhostar}{\rho^{\star}}
\newcommand{\Bdag}{B^{\dagger}}
\newcommand{\Bddag}{B^{\ddagger}}

\newcommand{\tsel}{\tau_{\mathrm{select}}}
\newcommand{\trep}{\tau_{\mathrm{repair}}}
\newcommand{\tnet}{\tau_{\mathrm{net}}}
\newcommand{\tlp}{\tau_{\mathrm{lp}}}          % cost of one relaxation solve
\newcommand{\PI}{P}                            % primal integral
\newcommand{\CPI}{P_{\alpha}}                  % confined primal integral

\newcommand{\pol}{\pi_{\theta}}
\newcommand{\heat}{M}                          % predicted heatmap
\DeclareMathOperator{\TopK}{TopK}
\DeclareMathOperator{\Sample}{Sample}

\newcommand{\Aloc}{\textup{(A1)}}
\newcommand{\Aden}{\textup{(A2)}}
\newcommand{\Adens}{\textup{(A2$'$)}}
\newcommand{\Areg}{\textup{(A3)}}
\newcommand{\Alip}{\textup{(A4)}}
\newcommand{\Bcontr}{\textup{(B1)}}
\newcommand{\Bplat}{\textup{(B1$'$)}}

\theoremstyle{plain}
\newtheorem{theorem}{Theorem}
\newaliascnt{proposition}{theorem}
\newtheorem{proposition}[proposition]{Proposition}
\aliascntresetthe{proposition}
\newaliascnt{lemma}{theorem}
\newtheorem{lemma}[lemma]{Lemma}
\aliascntresetthe{lemma}
\newaliascnt{corollary}{theorem}
\newtheorem{corollary}[corollary]{Corollary}
\aliascntresetthe{corollary}
\theoremstyle{definition}
\newaliascnt{definition}{theorem}
\newtheorem{definition}[definition]{Definition}
\aliascntresetthe{definition}
\theoremstyle{remark}
\newaliascnt{remark}{theorem}
\newtheorem{remark}[remark]{Remark}
\aliascntresetthe{remark}

\crefname{theorem}{Theorem}{Theorems}
\crefname{proposition}{Proposition}{Propositions}
\crefname{lemma}{Lemma}{Lemmas}
\crefname{corollary}{Corollary}{Corollaries}
\crefname{definition}{Definition}{Definitions}
\crefname{remark}{Remark}{Remarks}

\theoremstyle{plain}
\newtheorem*{proposition*}{Proposition}

\newcommand{\conditional}{\textnormal{\scriptsize\textsc{[conditional]}}}

\title{Recasting the Destroy Step of Large Neighborhood Search as Dense Segmentation}

\author{\textbf{Yang Liu}\textsuperscript{1},
\textbf{Yulin Huang}\textsuperscript{1},
\textbf{Jianshen Zhang}\textsuperscript{1},
\textbf{Yongzhi Qi}\textsuperscript{1}$^{\dagger}$
\\[0.6em]
\textsuperscript{1}Supply Chain Tech Team Y, JD.com\\
$^{\dagger}$Corresponding author.
}

\iclrfinalcopy % Uncomment for camera-ready version, but NOT for submission.

\begin{document}

\maketitle

\begin{abstract}
Large neighborhood search improves an incumbent by releasing selected variables and repairing the resulting subproblem under a time limit. FOVEA casts variable selection as dense semantic segmentation on a fixed $128\times128$ canvas. A small U-Net reads twelve semantic channels and predicts cell scores, which are decoded into a variable set subject to a fixed variable budget. Family-specific layouts connect the search state to the canvas, while one set of network weights serves four problem families. The FP32 network input occupies $786$\,kB across instance sizes. Our analysis bounds constraint coupling for spatially coherent regions, gives a rank certificate for rounds with zero possible improvement, and bounds the coupling advantage attainable on an expander family. On the tested families and hardware, FOVEA trails the strongest graph encoder at $10^{4}$ variables and overtakes it near $1.5\times10^{4}$; the cost model provides an approximate crossover estimate. At larger sizes, where the graph-encoder baselines exceed device memory, FOVEA reduces the primal integral by $12$ to $16\%$ relative to the best runnable baselines. On the expander family, FOVEA performs comparably to random selection, with a coupling advantage close to one.
\end{abstract}
\section{Introduction}
\label{sec:intro}

% Figure 1. Three bands: the constraint, the reformulation, the predictions.
% The panels are rendered from one illustrative instance built to satisfy the
% hypotheses of Section 3. All text is set here so that it matches the body font.
\suppressfloats[t]
\begin{figure}[t]
\centering
\begin{tikzpicture}[
  x=1cm, y=1cm,
  band/.style={fill=black!3, rounded corners=2pt, draw=none},
  rule/.style={draw=black!18, line width=0.4pt},
  tag/.style={fill=fovblue, text=white, font=\bfseries\scriptsize,
              inner sep=0pt, minimum width=3.4mm, minimum height=3.4mm,
              rounded corners=1pt},
  title/.style={font=\scriptsize\bfseries, anchor=west, inner sep=0pt},
  lab/.style={font=\tiny, align=center, inner sep=1pt},
  labl/.style={font=\tiny, align=left, inner sep=0pt, anchor=north west},
  pane/.style={inner sep=0pt, outer sep=0pt},
  data/.style={fill=white, draw=black!30, line width=0.35pt,
               inner sep=0.6pt, outer sep=0pt},
  box/.style={draw=black!55, rounded corners=1.5pt, align=center,
              font=\tiny, inner xsep=1pt, inner ysep=2pt,
              minimum height=6.5mm, text width=12mm},
  once/.style={box, dashed, draw=black!45},
  flow/.style={-{Latex[length=1.5mm,width=1.2mm]}, draw=black!65, line width=0.5pt},
  back/.style={flow, dashed},
  pred/.style={font=\tiny, align=center, inner sep=0pt, text width=30mm,
               anchor=north},
  hit/.style={pred, text=fovblue}
]

% =====================================================================
% Band A: the asset the layout creates, and the constraint it removes
% =====================================================================
\fill[band] (0,0) rectangle (13.9,-3.15);
\node[tag] at (0.30,-0.30) {A};
\node[title] at (0.55,-0.30)
  {Local layouts connect coupling to region boundaries};

\node[data] (spyraw) at (1.35,-1.75)
  {\includegraphics[height=15.5mm]{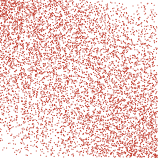}};
\node[data] (spylay) at (4.05,-1.75)
  {\includegraphics[height=15.5mm]{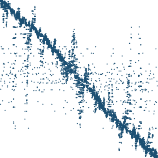}};
\draw[flow] (spyraw) -- node[lab, above=0.2mm] {$\lay$}
  node[lab, below=0.2mm] {MILP setup,\\$O(\nnz(A))$} (spylay);
\node[lab, anchor=north] at (1.35,-2.58) {rows of $A$,\\arbitrary order};
\node[lab, anchor=north] at (4.05,-2.58) {same $A$, indexed\\by canvas cells};

\draw[rule] (5.30,-0.70) -- (5.30,-2.95);

\node[labl, text width=37mm] at (5.55,-0.78)
  {Under the locality and regularity assumptions, the cut of a full region
   preimage scales with its \emph{perimeter}:
   $O(\kappa m\sqrt{sk\Bpar}/n)+\eta m$, compared with
   $\Omega(\seff km/n)$ for a random set of the same size (\Cref{prop:iso},
   \Cref{lem:random}). This motivates spatial neighborhood selection.};

\draw[rule] (9.55,-0.70) -- (9.55,-2.95);

\node[pane] (wall) at (10.65,-1.90)
  {\includegraphics[height=12mm]{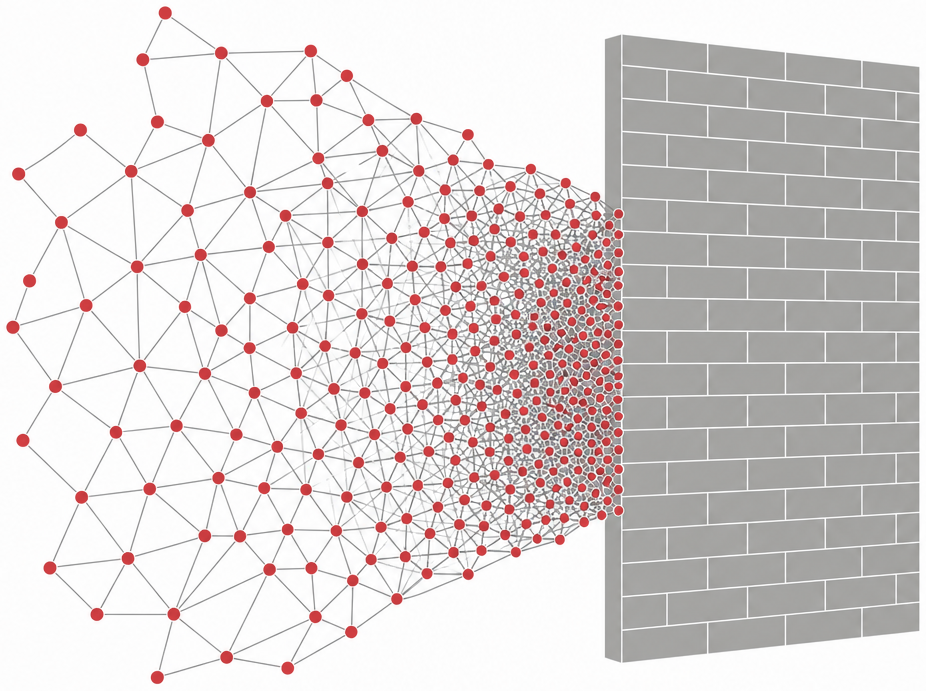}};
\node[labl, text width=21mm] at (11.60,-0.90)
  {Graph-encoder activations grow with $n$. The fixed canvas keeps the
   policy input size constant.};
\node[lab, anchor=north] at (10.65,-2.60) {activations $\propto n$};

% =====================================================================
% Band B: one round, at a size that does not depend on n
% =====================================================================
\fill[band] (0,-3.35) rectangle (13.9,-6.70);
\node[tag] at (0.30,-3.65) {B};
\node[title] at (0.55,-3.65)
  {Dense segmentation at any $n$: the network reads $786$\,kB};

\node[once] (lay) at (0.95,-5.15) {layout $\lay$\\\emph{setup}};
\node[pane] (can) at (2.95,-5.15)
  {\includegraphics[height=14mm]{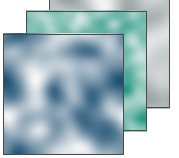}};
\node[box, text width=11mm] (net) at (5.40,-5.15) {U-Net $\pol$\\$1$--$5$M};
\node[data] (heat) at (7.75,-5.15)
  {\includegraphics[height=14mm]{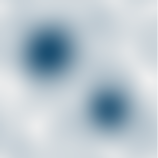}};
\node[data] (mask) at (10.10,-5.15)
  {\includegraphics[height=14mm]{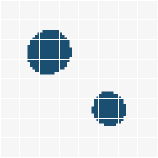}};
\node[box, text width=13mm] (rep) at (12.45,-5.15)
  {repair $\sub{S}$\\time limit $\trep$};

\draw[flow] (lay) -- (can);
\draw[flow] (can) -- (net);
\draw[flow] (net) -- (heat);
\draw[flow] (heat) -- (mask);
\draw[flow] (mask) -- (rep);

\node[lab, anchor=south] at (2.95,-4.35) {canvas $\canvas^{(t)}$,\\$128^{2}\times12$};
\node[lab, anchor=south] at (5.40,-4.35) {$O(HW)$,\\free of $n$};
\node[lab, anchor=south] at (7.75,-4.35) {field $\heat$\\over cells};
\node[lab, anchor=south] at (10.10,-4.35) {$\Sample$, $\TopK$:\\at most $k$ variables};
\node[lab, anchor=south] at (12.45,-4.35) {new incumbent\\$\inc$};

\draw[back] (rep.south) -- ++(0,-0.75) -| (can.south);
\node[lab, fill=black!3, inner sep=1.5pt] at (7.20,-6.28)
  {Changed-cell updates, $O(|S|+HW)$ with bounded incidence degree};

% =====================================================================
% Band C: what the analysis risks, and what the measurement says
% =====================================================================
\fill[band] (0,-6.90) rectangle (13.9,-8.95);
\node[tag] at (0.30,-7.20) {C};
\node[title] at (0.55,-7.20)
  {Analysis, experimental results and a certificate illustration};

\foreach \i/\claim/\meas in {%
  0/{coupling envelope: $\rhostar=1/s$ (\Cref{cor:rhostar})}/{peak at $\rho=0.10$, $s=9.6$}%
 ,1/{a crossover at $\nstar$, put at \mbox{$10^{4.1}$--$10^{4.4}$} (\Cref{prop:critical})}/{the curves cross at $10^{4.19}$}%
 ,2/{$\rank(A_{E(S),S})=k$ forbids improvement (\Cref{prop:rank})}/{$4000$ schematic points; full-rank gain $0$ by construction}%
 ,3/{expansion limits coupling advantage (\Cref{prop:expander})}/{$1.01\times$, $p=0.61$}}
{
  \node[pred] at ({1.925+3.35*\i},-7.50) {\claim};
  \draw[-{Latex[length=1.2mm,width=1.0mm]}, draw=fovblue!75, line width=0.5pt]
    ({1.925+3.35*\i},-8.14) -- ++(0,-0.22);
  \node[hit] at ({1.925+3.35*\i},-8.42) {\meas};
}
\foreach \i in {1,2,3}
  \draw[rule] ({0.25+3.35*\i},-7.42) -- ({0.25+3.35*\i},-8.80);

\end{tikzpicture}
\caption{FOVEA overview. \textbf{(A)} A precomputed local layout relates
the coupling of a full region preimage to its perimeter.
\textbf{(B)} Each round rasterises the search state into a fixed-size tensor,
predicts a field over cells, truncates the variables under the mask to at most
$k$, and solves the resulting repair formulation within the time budget.
The one-time layout shown here applies to MILP and routing. Changed cells are updated incrementally (\Cref{sec:incremental}).
\textbf{(C)} Coupling, crossover, and rank predictions are paired with
experimental outcomes and the rank schematic, including the expander control. Panels (A) and (B) use an
illustrative instance satisfying the hypotheses of \Cref{sec:analysis}.
Panel (C) summarises the text; its rank example uses generated points.}
\label{fig:hero}
\end{figure}
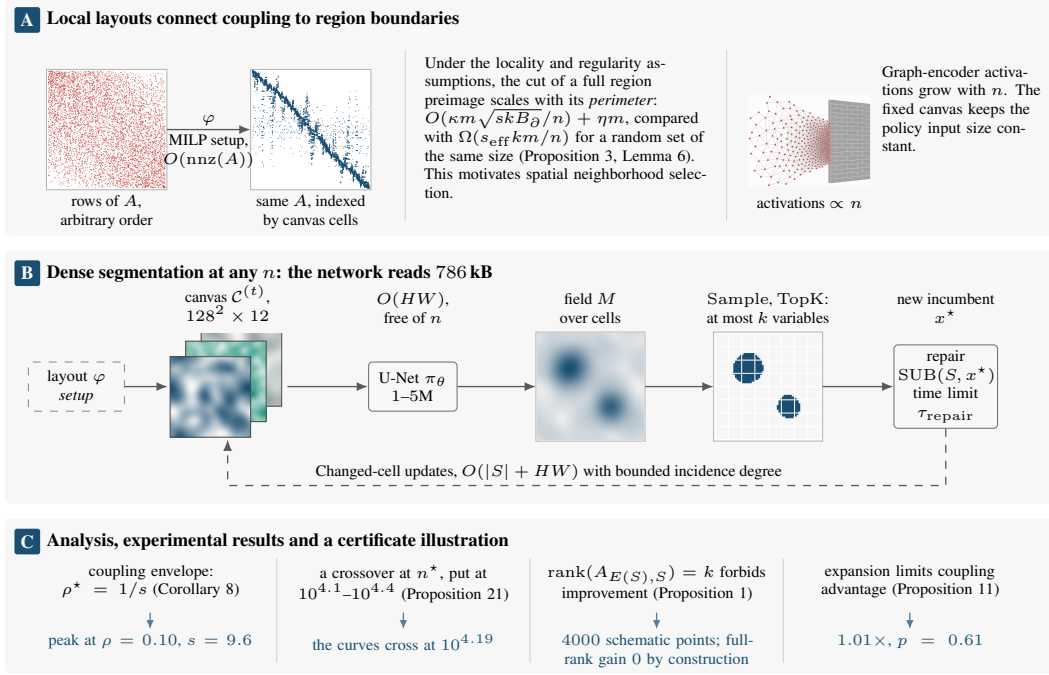

Mixed integer programming models routing, scheduling, network design and resource allocation. At large scales, decisions often require a feasible solution before branch-and-bound closes the optimality gap. Large neighborhood search \citep{shaw1998lns,pisinger2007alns} maintains a feasible incumbent, releases a small subset of its variables, solves the resulting subproblem under a time limit, and accepts improvements. The destroy step chooses one of the $\binom{n}{k}$ subsets of $k$ variables; a solver then repairs the selected neighborhood. Classical rules \citep{danna2005rins,fischetti2003localbranching,pisinger2007alns} combine hand-designed signals from the instance and incumbent. Learned operators adapt the selection rule to the search state.

Learned destroy operators often use representations indexed by the instance \citep{song2020generallns,wu2021rllns,sonnerat2021neurallns,huang2023cllns,li2021l2d}. For the full-graph encoders evaluated here, a message-passing layer costs $O(\nnz(A))$, and the graph and activations occupy device memory. On the tested families and hardware, these encoders reach the memory limit beyond a few times $10^{4}$ variables. This motivates a selector whose neural input and output dimensions remain fixed as the instance grows.

A spatial grid provides such a representation domain. Visual representations already support combinatorial optimisation through height maps for bin packing, rendered constraint matrices and images of tours \citep{zhao2021bpp,ling2021milp,steever2022mip,vnsolver2023,gimf2025,vitsp2025}. These applications include constructing solutions with a decoder or graph encoder. For destroy selection, the spatial output directly identifies candidate neighborhoods: a mask selects cells, and a decoder maps their entities to variables for repair.

Two close visual approaches highlight the roles of resolution and action geometry. \citet{gimf2025} fuse image and graph features and increase the total pixel count linearly with instance size to support scale transfer. \citet{vitsp2025} select axis-aligned bounding boxes on a rendered tour and repair their interiors exactly. FOVEA uses a fixed-resolution canvas and a mask over its cells, keeping the number of neural outputs independent of instance size. An entity-to-variable map provides the interface to each repair problem. This separates the spatial action from the family-specific definition of a released variable. The resolution and mask-shape comparisons (\Cref{tab:ablation-results}) measure the associated quality tradeoffs, including scale-dependent resolution and rectangular selections.

FOVEA rasterises the search state onto a fixed canvas and predicts cell scores with a U-Net. Entity mapping and budgeted decoding turn the selected cells into the variables released for repair. Spatial adjacency makes perimeter penalties, morphological post-processing and connected-component analysis available to the selector. Family-specific layouts provide the spatial interface: routing and general mixed integer programs reuse precomputed layouts, while scheduling updates its layout with the incumbent. Incremental rasterisation accounts for variable, constraint and periodic relaxation updates (\Cref{sec:incremental}).

The shared canvas protocol connects one segmentation policy to variable-level repair across routing, mixed integer programming, job shop scheduling and bin packing. Related cross-problem interfaces use heterogeneous graph types, token schemes or attribute vectors \citep{drakulic2024goal,zhou2024mvmoe,unico2025}; \Cref{app:related} compares these interfaces with the canvas. Our experiments separate the effects of spatial layout and learning: layout accounts for most of the gain on set cover and routing, and learning improves on the same-layout control. The analysis relates region boundaries to constraint coupling, gives a rank certificate for zero-improvement neighborhoods, and identifies coupling limits on expanders. These results motivate the layout, budget and resolution comparisons, including a resolution schedule based on \citet{gimf2025}.

\section{Preliminaries}
\label{sec:prelim}

We consider mixed integer linear programs $\min\{c^{\top}x:x\in\feas\}$ over $\feas=\{x:A_{j}x=b_{j}\ (j\in E),\ A_{j}x\le b_{j}\ (j\in\cons\setminus E),\ x_{i}\in\integers\ (i\in\mathcal{I}),\ l\le x\le u\}$, with variable set $\vars$ and constraint set $\cons$. The equality rows $E\subseteq\cons$ define the rank certificate in \Cref{prop:rank}; $\thetaE=|E|/m$ is their share. The support of constraint $j$ is $\supp(A_{j})=\{i:A_{ji}\neq0\}$ with size $s_{j}$, the mean support is $s=\nnz(A)/m$, and $d=\nnz(A)/n$ is the mean number of constraints per variable. Routing and scheduling enter this framework through their integer formulations and the entity mapping in \Cref{sec:entity}.

Large neighborhood search maintains a feasible incumbent $\inc$ and repeats two steps. The \emph{destroy} step chooses $S\subseteq\vars$ with $|S|=k$. The \emph{repair} step solves, under a time limit,
\begin{equation}
\sub{S}:\quad \min_{x}\;c^{\top}x
\quad\text{s.t.}\quad x\in\feas,\qquad x_{i}=\inc_{i}\;\;\forall i\notin S,
\label{eq:sub}
\end{equation}
and replaces $\inc$ if the result improves it. We write $\gain(S)=c^{\top}\inc-\min\{c^{\top}x:x\in\feas(\sub{S})\}$ for the improvement the best possible repair of $S$ would deliver.

Coupling helps characterise the freedom available in a destroy set. It counts the constraints straddling released and fixed variables,
\begin{equation}
\cut{S}=\bigl|\{\,j\in\cons \;:\; \supp(A_{j})\cap S\neq\emptyset \;\wedge\; \supp(A_{j})\setminus S\neq\emptyset \,\}\bigr|.
\label{eq:cut}
\end{equation}
A straddling constraint enters \eqref{eq:sub} as a constraint on $x_{S}$, with the fixed variables contributing a constant. \Cref{sec:analysis} characterises the resulting freedom through the rank of the equality submatrix: full column rank makes $\inc$ the unique feasible point and certifies zero possible improvement.

We evaluate solution quality over time using the primal integral. Let $\opt$ be a reference value and $\gamma(x)$ the primal gap: zero if $c^{\top}x=\opt$, one if the two have opposite signs, and $|c^{\top}x-\opt|/\max\{|c^{\top}x|,|\opt|\}$ otherwise. The case split keeps $\gamma\in[0,1]$ when incumbent and reference straddle zero, which makes the integral comparable across instances. Following \citet{berthold2013primal}, the primal gap function $p(t)$ is $1$ before the first feasible solution and $\gamma(\text{incumbent at }t)$ afterwards, and the primal integral is $\PI(T)=\int_{0}^{T}p(t)\,dt$. \Cref{app:metrics} also defines the confined primal integral of \citet{berthold2021confined}, an exponentially discounted measure with a finite infinite-horizon limit. Comparing it with $\PI$ assesses sensitivity to the evaluation horizon. Since $p(0)=1$, $\PI(T)$ contains a term equal to the time $t_{0}$ of the first feasible solution, which belongs to the shared initialisation rather than to any destroy policy; \Cref{fig:convergence} therefore draws that segment explicitly, and we keep $t_{0}$ and $\PI(T)-t_{0}$ apart.

We study a fixed absolute variable budget $k\in[10^{2},10^{3}]$, independent of $n$. The corresponding destroy fraction is $\rho=k/n$, which decreases as the instance grows. Holding a fraction fixed would instead give $k=\Theta(n)$: for example, $n=10^{6}$ and $\rho=0.1$ release $10^{5}$ variables for repair. \Cref{app:regime} describes the computational and coverage properties of the fixed-budget regime. Under the assumptions of \Cref{cor:rhostar}, the geometric lower bound on coupling advantage peaks at $\rhostar=1/s$. At fixed $k$, instances with $n>ks$ lie on the lower-$\rho$ side of this peak, where the bound rises with $\rho$.

\section{FOVEA}
\label{sec:method}

FOVEA evaluates a segmentation policy on a fixed-size canvas. Band B of \Cref{fig:hero} shows one round: rasterise the search state, predict cell scores, sample cells, decode variables within the budget, and repair the selected neighborhood. \Cref{alg:fovea} gives the search loop. The entity layer defines how each problem family connects its spatial representation to the variables of the repair problem.

\subsection{The entity layer}
\label{sec:entity}

The layout $\lay$ maps \emph{variables} to canvas positions. In routing, coordinates belong to customers while the decision variables are edges $x_{ij}$. An entity layer connects these two representations and defines the variable set used by $\cut{S}$ and \Cref{sec:analysis}. Each family supplies an entity set $\ents$, a row-stochastic placement $P\in\reals^{|\ents|\times HW}$ spreading each entity's unit mass over cells, and an association $\ent:\vars\to2^{\ents}$ recording the entities referenced by each variable. An entity is selected when its mass inside the chosen cells reaches a threshold $\vartheta$; a variable is released when all of its entities are selected:
\begin{equation}
\ents_{\mathrm{sel}}=\Bigl\{e\in\ents:\textstyle\sum_{c\in\Psel}P_{ec}\ge\vartheta\Bigr\},
\qquad
S=\bigl\{\,i\in\vars:\ent(i)\subseteq\ents_{\mathrm{sel}}\,\bigr\}.
\label{eq:entity-decode}
\end{equation}
Requiring \emph{all} entities selects an edge only when both endpoints are selected. For general programs the entity is the variable and the association is the identity; for routing the entity is the customer and $S$ contains candidate edges among selected customers (\Cref{tab:entities} gives all four instantiations). Routing uses a sparsified candidate edge set to support locality: on a complete graph, degree constraints span the canvas (\Cref{rem:eta}).

\subsection{Layout}
\label{sec:layout}

The layout aims to make constraints spatially local, as required by \Aloc{}. Routing uses native coordinates. For general programs, the constraint matrix supplies a two-dimensional sparsity pattern indexed by (constraint, variable). We reorder its rows and columns to concentrate nonzeros near the diagonal. A variable co-occurrence graph adds $\binom{s_{j}}{2}$ edges for a constraint of support $s_{j}$. To retain sparse processing, we run reverse Cuthill-McKee \citep{cuthill1969rcm} on the bipartite adjacency $\bigl(\begin{smallmatrix}0&A\\A^{\top}&0\end{smallmatrix}\bigr)$ at $O(\nnz(A))$. The resulting order induces both permutations and brings variables sharing constraints closer along the column axis. The ablation design includes a one-dimensional canvas, motivated by the dimension-dependent bound in \Cref{cor:dimension}. An index-order control isolates the contribution of spatial layout.

\subsection{The canvas}
\label{sec:canvas}

At round $t$ the state is rasterised into $\canvas^{(t)}\in\reals^{H\times W\times K}$ with $H=W=128$ and $K=12$. Cell $(p,q)$ owns the variables $\Vpq=\{i:\lay(i)\in\cell(p,q)\}$. Each channel aggregates a numerical quantity over $\Vpq$, covering instance geometry, current state, dual information or search history (\Cref{tab:channels}). The network receives the aggregated channel values as FP32 numbers. Several draw on signals used by classical heuristics: incumbent--relaxation discrepancy in RINS \citep{danna2005rins}, reduced cost in variable fixing, and pseudocost in branching. The policy learns how to combine these signals when selecting a region.

The network's tensor is one part of the search state. The rasteriser maintains per-cell aggregates, and a host-side index maps cells back to variables in the original problem. These structures connect the fixed neural representation to variable-level repair: the policy reuses the same input shape across scales, while the index retains the instance-specific information needed for decoding (\Cref{app:memory}).

For a fixed layout, the twelve channels have three update schedules: four static channels are computed once, four state-dependent channels use incremental updates at $O(|S|+HW)$ for grid and variable state, and four dual channels are refreshed every $\Rlp$ rounds. Selection time includes the amortised refresh cost $\tlp/\Rlp$ and the constraint-update cost specified in \Cref{sec:incremental}. Routing and constraint-programming sub-solvers supply no LP relaxation, so their dual channels use the neutral values in \Cref{app:method} and incur no LP refresh cost. Mixed integer programming includes that additional cost.

Across three decades in $n$, a fixed canvas receives occupancy counts that also span three decades. We use within-instance normalisation and GroupNorm \citep{wu2018groupnorm} throughout to control this input variation. The reported resolution comparisons use normalised inputs and examine the scale dependence identified by \citet{gimf2025}. The ablation design also specifies a separate normalisation study.

Count channels enter as $\log(1+v)$ rescaled by the within-instance $99$th percentile of the same transform, which handles both the three orders of magnitude of variation in $n$ and the within-canvas non-uniformity of clustered instances. Mean-valued channels enter as within-instance robust $z$-scores, using the median and the median absolute deviation so that a few extreme cells do not set the scale. History channels are normalised by the current round index and by the largest improvement seen so far, so that both are bounded and comparable across a run. All statistics are computed per instance.

\subsection{Policy and two-stage decoding}
\label{sec:policy}

The policy $\pol:\reals^{H\times W\times K}\to[0,1]^{H\times W}$ is a U-Net \citep{ronneberger2015unet} with four downsampling stages and one to five million parameters. We train it for the twelve semantic channels and target latency below $50$\,ms on one CPU core to support CPU-only deployment.

Decoding first selects cells and then variables. $\Sample(\heat)$ perturbs the log-heatmap with Gumbel noise at scale $\varepsilon$ and visits cells in decreasing perturbed order until they hold $\beta k$ variables. At unit noise scale, this samples without replacement proportional to $\heat$ \citep{kool2019stochasticbeam}; the scale controls sampling concentration. Morphological closing reduces boundary fragmentation, motivated by \Cref{cor:blocks}. If the positively responding cells are exhausted first, the round uses the available cells and a smaller budget.

The variable stage applies the budget to the variables associated with the sampled cells. The default rule consumes whole cells in decreasing cell score and truncates the final cell by per-variable score. When a relaxation is available, the score is the absolute difference between incumbent and relaxation values; routing and scheduling use family-specific quantities. Cell scores have dimension $HW$ for every $n$, while the variable stage provides finer selection within cells. \Cref{app:truncation} describes the default rule and the score-only alternative.

The coupling bounds concern the full region preimage $\SR$. Budgeted decoding returns a subset, and the full-preimage bound does not automatically apply after truncation. \Cref{app:truncation} describes the proposed comparison of truncation rules. At large $n$, when a cell can contain thousands of variables, the optional refinement in \Cref{app:foveation} re-rasterises the selected region and reapplies $\pol$.

\subsection{Incremental rasterisation and training}
\label{sec:incremental}

\begin{algorithm}[t]
\caption{FOVEA search with a fixed variable budget}
\label{alg:fovea}
\small
\begin{algorithmic}[1]
\Require Instance $(\feas,c)$, feasible incumbent $\inc$, trained policy $\pol$, budget $k$, wall-clock deadline
\State Build the family layout, variable index and initial canvas $\canvas$.
\While{time remains}
  \State If $\inc$ changed, update any incumbent-dependent layout; refresh affected canvas state.
  \State Refresh available relaxation channels when due (\Cref{sec:canvas}).
  \State $\heat\gets\pol(\canvas)$; $\Psel\gets\Sample(\heat)$, including morphological processing.
  \State $S\gets\operatorname{BudgetDecode}(\Psel,k)$ using the rule in \Cref{sec:policy}.
  \State If time remains, solve $\sub{S}$ for at most $\min\{\trep,\text{time left}\}$.
  \State Accept a returned feasible solution if it improves $c^{\top}\inc$.
  \State Record the round's outcome and update the search history.
\EndWhile
\State \Return $\inc$
\end{algorithmic}
\end{algorithm}

Between rounds only variables in $S^{(t)}$ change value, so updates visit the cells meeting $\lay(S^{(t)})$. Rebuilding the grid costs $O(HW)$, and redistributing placement mass costs $O(k\max_{i}|\ent(i)|)$. The tightness channel additionally revisits incident constraints at cost $O(\sum_{j\in E(S^{(t)})}s_{j})$. With bounded constraint degrees and support sizes, these updates are independent of $n$; dense rows increase their cost beyond the grid-and-variable term $O(HW+k)$. The first rasterisation costs $O(n)$ and takes seconds at $n=10^{6}$. We report this setup separately and measure incremental update time across instance sizes.

Training imitates a local branching oracle \citep{fischetti2003localbranching,sonnerat2021neurallns}, uses hindsight relabelling from diverse behaviour policies, and optionally applies REINFORCE \citep{williams1992reinforce} to improvement per unit time (\Cref{app:method}). Positive labels cover the variables that a sub-solve actually changes, focusing supervision on effective moves. All families share one $\pol$ without a family embedding, so held-out families use the same trained parameters and input protocol.

\subsection{Analysis}
\label{sec:analysis}

The analysis relates region boundaries to constraint coupling, identifies neighborhoods with zero possible improvement, and bounds coupling advantage on expander instances. Improvement in the objective also depends on the repair problem. The geometric results use locality \Aloc{}, density \Aden{} and region regularity \Areg{}; \Cref{app:proofs} gives their definitions and the assumptions for each result. The tag \conditional{} identifies results that use additional distributional, geometric or cost-model assumptions. \Cref{tab:predictions} distinguishes associated empirical criteria from the illustrative rank panel.

Coupling depends on the boundary of the selected region. Write $\SR=\lay^{-1}(\region)$ for its variable preimage. At scaling $\delta=\cdelta\sqrt{s/n}$ and $|\region|=k/n$, under the locality, density, regularity and budget conditions of \Cref{prop:iso}, the coupling bound is $O(\kappa m\sqrt{sk\Bpar}/n)+\eta m$. Its boundary term grows with the square root of the released mass; $\eta m$ accounts for constraints that are nonlocal under the chosen layout. Combining this bound with the random-selection lower bound in \Cref{lem:random} gives the geometric factor $\kappa^{-1}\min\{\sqrt{sk},\,n/\sqrt{sk}\}$. Under the additional support and defect assumptions of \Cref{cor:rhostar}, this lower bound on coupling advantage has peak $\Omega(\sqrt{n}/\kappa)$ at $\rhostar=1/s$. For the set cover family with $s\approx10$, this places the geometric peak near $\rhostar\approx0.1$. The random baseline uses a truncated support mean, and the bound applies to $\SR$ before decoder truncation (\Cref{app:proofs}).

\begin{proposition}[Degree-of-freedom certificate]
\label{prop:rank}
Let $E$ be the equality rows and $E(S)$ those meeting $S$. If $\rank(A_{E(S),S})=|S|=k$, then $\inc$ is the unique feasible point of $\sub{S}$ and the round cannot improve. Rank deficiency, $\rank(A_{E(S),S})<k$, is therefore necessary for improvement. Since $\rank(A_{E(S),S})\le\cutE{S}+\inE{S}$, fewer than $k$ incident equality rows guarantee rank deficiency: $\cutE{S}+\inE{S}<k$ implies $\rank(A_{E(S),S})<k$.
\end{proposition}

A row count of at least $k$ can coexist with deficient rank and possible improvement; the certificate itself uses $\rank(A_{E(S),S})=k$. At budgets $k$ in the hundreds, the rank test takes milliseconds. The sufficient condition also motivates an online budget rule: enlarge $k$ until the row count falls below it. In the parameter setting of \Cref{app:rank-substitution}, this occurs above a budget threshold for spatial regions, while random sets fail the condition over the analysed range. The certificate applies to equality rows, including routing degree constraints, packing assignment rows and many MIPLIB instances. The four synthetic families contain only inequalities; there the budget rule uses the tight-inequality heuristic of \Cref{rem:tight}.

\paragraph{Coupling limits on expanders.}
For a left $d$-regular lossless-expander incidence graph satisfying the uniform-support and neighborhood-size conditions of \Cref{prop:expander}, every selection cuts at least $(1-2\epsilon)$ times the expected random coupling. For $\epsilon<1/2$, the attainable coupling advantage is at most $(1-2\epsilon)^{-1}$, uniformly over admissible $S$ on each instance. Even with $s=O(1)$ and $d=O(1)$, spatial layouts must satisfy both this lower bound and the geometric upper bound wherever their hypotheses overlap. These bounds constrain the achievable locality and density parameters; their precise compatibility is discussed in \Cref{app:proofs}. Explicit expander constructions \citep{capalbo2002expanders} provide the test family. We measure its coupling advantage and primal integral separately. Five further results in \Cref{app:proofs} guide the design and cost accounting: boundary fragmentation motivates the perimeter loss and sampler post-processing (\Cref{cor:blocks}); dimension-dependent isoperimetry motivates the one-dimensional canvas comparison (\Cref{cor:dimension}); the resolution bound predicts saturation under regularity assumptions, with a fitted constant (\Cref{cor:gstar}); and the critical-size and memory analyses specify the scaling tests (\Cref{prop:critical,app:memory,sec:experiments}).

\section{Experiments}
\label{sec:experiments}

We measure search quality by the primal integral, charging both selection and repair time. Routing and mixed integer programming provide training data, while scheduling and bin packing test zero-shot transfer of the shared weights. A fifth family, set cover with a lossless-expander incidence graph, tests the coupling limit in \Cref{prop:expander}. The main experiments use a $200$\,s single-threaded horizon and five seeds. A region-growing control shares the canvas, layout, truncation rule, budget schedule, sub-solver and seeds, isolating the contribution of the network. \Cref{app:experiments} gives the splits, seventeen baselines, criteria fixed before the runs, statistical tests and separate MIPLIB protocol. It also describes a thirty-five-study ablation design, with results for five studies in \Cref{tab:ablation-results}.

\begin{table}[t]
\caption{Synthetic mixed integer programs: primal integral in seconds over a $200$\,s compute-charged horizon. Lower is better; best values are bold and FOVEA is shaded. The last row gives FOVEA's percentage improvement over the best runnable baseline. \textsc{oom} denotes an encoder exceeding device memory; a dash denotes a sub-solve that returned no result within the horizon. \Cref{tab:milp} separately lists the published results of \citet{huang2023cllns}, obtained with a $60$\,min protocol at smaller sizes. The other two families appear in \Cref{tab:milp-full}.}
\label{tab:milp-main}
\centering
\footnotesize
% Generated table body. Edit the results data, not this file.
\begin{tabular}{@{}lccc@{\hspace{1.15em}}ccc@{}}
\toprule
Method & \multicolumn{3}{c}{SC} & \multicolumn{3}{c}{CA} \\
\cmidrule(lr){2-4}\cmidrule(lr){5-7}
& $10^{4}$ & $10^{5}$ & $10^{6}$ & $10^{4}$ & $10^{5}$ & $10^{6}$ \\
\midrule
SCIP, default & 58.4 & 128.6 & 171.2 & 74.9 & 149.3 & 184.1 \\
SCIP, aggressive & 46.2 & 112.4 & 158.9 & 63.1 & 134.7 & 176.3 \\
Random LNS & 34.7 & 71.3 & 96.4 & 68.2 & 103.8 & 138.9 \\
RINS & 31.9 & 68.9 & 94.8 & 59.4 & 98.2 & 132.4 \\
Local branching & 62.8 & --- & --- & 88.3 & --- & --- \\
LB-RELAX & 28.6 & 64.2 & 89.1 & 52.7 & 92.4 & 127.8 \\
IL-LNS & 25.3 & \textsc{oom} & \textsc{oom} & 47.6 & \textsc{oom} & \textsc{oom} \\
CL-LNS & \textbf{21.4} & \textsc{oom} & \textsc{oom} & \textbf{41.2} & \textsc{oom} & \textsc{oom} \\
\midrule
Region growing & 24.1 & 55.7 & 78.5 & 49.8 & 84.6 & 115.2 \\
\rowcolor{gray!12}FOVEA & 22.0 & \textbf{48.3} & \textbf{66.2} & 43.5 & \textbf{74.1} & \textbf{99.6} \\
\addlinespace[1pt]
\emph{gain over best runnable} & -2.8 & \textcolor{fovblue}{\textbf{+13.3}} & \textcolor{fovblue}{\textbf{+15.7}} & -5.6 & \textcolor{fovblue}{\textbf{+12.4}} & \textcolor{fovblue}{\textbf{+13.5}} \\
\bottomrule
\end{tabular}

\end{table}

\paragraph{Scaling and computational cost.}
At $10^{4}$ variables, FOVEA has a $2.8\%$ higher primal integral than CL-LNS on set cover and a $5.6\%$ higher integral on auctions (\Cref{tab:milp-main}). At this scale, the $128^{2}$-cell representation trades direct variable scoring for aggregated state information. At $10^{5}$ and $10^{6}$ variables, both graph encoders exceed device memory. FOVEA lowers the primal integral relative to the best runnable baselines by $13.3\%$ and $15.7\%$ on set cover, and $12.4\%$ and $13.5\%$ on auctions. Each large-scale comparison has $p<0.05$ and concerns methods that fit the stated hardware budget.

\begin{figure}[t]
\centering
\includegraphics{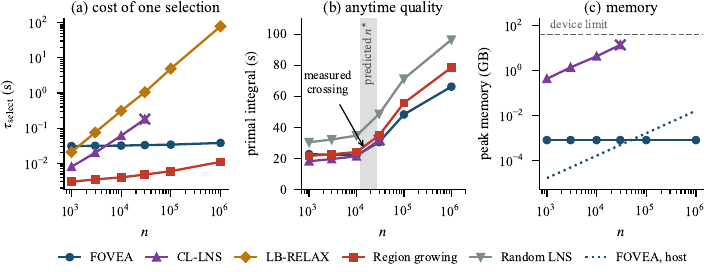}
\caption{Set cover across three decades of instance size, with budget $k=500$ fixed independently of $n$. Half-decade samples below $10^{5}$ bracket the crossing in (b). A cross marks the last size completed within the $40$\,GB device budget. The band is the $\nstar$ interval from \Cref{prop:critical} using the constants in \Cref{tab:timing}. The dotted line in (c) gives FOVEA's $O(n)$ host-side state.}
\label{fig:scaling}
\end{figure}

Selection time in \Cref{fig:scaling}(a) varies by at most $23\%$ across three decades, rising from $31$ to $38$\,ms as incremental updates touch more constraints per released variable. The graph encoder grows approximately in proportion to instance size until its memory limit. Region growing has lower selection cost throughout, by a factor of ten at $n=10^{3}$; the learned policy provides the quality improvement examined below. Below $n\approx5\times10^{3}$, FOVEA's selection step also costs more than the graph encoder's. Overall performance reflects both this cost and neighborhood quality, including the CL-LNS advantage at $10^{4}$ in \Cref{tab:milp-main}. The constants in \Cref{tab:timing} place $\nstar$ between $10^{4.1}$ and $10^{4.4}$, using the seed spread of quality ratio $r=1.03$ and a $15\%$ tolerance on the fitted encoder slope. The curves cross at $10^{4.19}$. The estimate of $r$ comes from compute-matched runs, independently of the primal-integral comparison. Because $\nstar$ equates bounds with different tightness, we interpret the estimate at factor-of-two precision (\Cref{rem:nstar}).

\begin{table}[t]
\caption{Routing. CVRPLIB X and XL report primal integral and gap to the best known value. Large sets, from AGS through Lazio, use the objective ratio to FILO2 at equal wall clock, consistently across the column because FILO2 supplies the reference solutions at the largest sizes. HGS, LKH-3 and FILO2 are re-run under the same $200$\,s limit. This truncates the first two relative to their published protocols in \Cref{tab:cvrp-x}; FILO2's short configuration completes its iteration budget and reproduces its published value. For solvers whose incumbent traces were not recorded, we report only endpoint gaps or objective ratios. A dash in the ratio column denotes no returned result on the million-customer instance within the deadline. Lower is better.}
\label{tab:cvrp-main}
\centering
\footnotesize
% Generated table body. Edit the results data, not this file.
\begin{tabular}{@{}lcc@{\hspace{1.15em}}cc@{\hspace{1.15em}}c@{}}
\toprule
Method & \multicolumn{2}{c}{X, $10^{2}$--$10^{3}$} & \multicolumn{2}{c}{XL, $10^{3}$--$10^{4}$} & \multicolumn{1}{c}{AGS through Lazio, to $10^{6}$} \\
\cmidrule(lr){2-3}\cmidrule(lr){4-5}\cmidrule(lr){6-6}
& $\PI$ & gap & $\PI$ & gap & ratio \\
\midrule
Random LNS & 18.7 & 2.41 & 31.4 & 4.82 & 1.212 \\
ALNS & 13.2 & 1.63 & 24.6 & 3.71 & 1.157 \\
L2D, $k{=}5$ & 10.4 & 1.18 & 23.8 & 3.16 & --- \\
L2D, $k{=}10$ & 9.6 & 1.04 & 21.3 & 2.88 & --- \\
LKH-3 & --- & 0.51 & --- & 1.94 & --- \\
HGS & --- & \textbf{0.19} & --- & 0.71 & --- \\
FILO2, short & --- & 0.34 & --- & \textbf{0.62} & 1.000 \\
\midrule
Region growing & 9.8 & 1.12 & 17.9 & 2.35 & 1.094 \\
\rowcolor{gray!12}FOVEA & \textbf{8.9} & 0.88 & \textbf{15.2} & 2.02 & \textbf{1.061} \\
\bottomrule
\end{tabular}

\end{table}

\begin{figure}[t]
\centering
\begin{minipage}[b]{0.50\textwidth}
\centering
\includegraphics{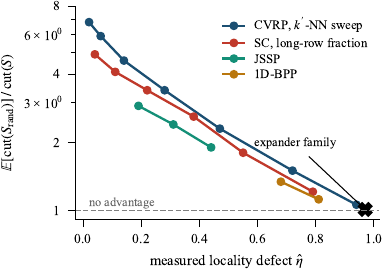}
\end{minipage}\hfill
\begin{minipage}[b]{0.46\textwidth}
\centering
\includegraphics{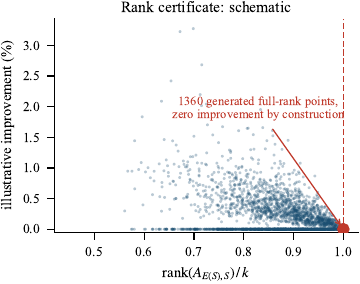}
\end{minipage}
\caption{Left: locality and coupling from the within-family sweeps associated with \Cref{prop:iso}. Routing varies the candidate-edge cutoff $k'$; set cover varies the fraction of long rows. These sweeps hold the problem family fixed. The cross marks the expander result, whose coupling advantage is limited by \Cref{prop:expander}. Right: a schematic of \Cref{prop:rank} containing $4000$ generated points. Of these, $1360$ satisfy $\rank(A_{E(S),S})=k$ and are assigned zero improvement by construction. The panel illustrates the certificate and provides no independent empirical test.}
\label{fig:rank}
\end{figure}

\paragraph{Routing: search trajectories and final quality.}
FOVEA has the lowest reported primal integral among the neighborhood search methods on CVRPLIB X and XL, lowering it relative to region growing by $9.2\%$ and $15.1\%$. Specialised solvers achieve better final objectives: HGS reaches a $0.19\%$ gap on CVRPLIB X compared with FOVEA's $0.88\%$, and FOVEA's objective averages $6.1\%$ above FILO2's across the large sets at the same deadline (\Cref{tab:cvrp-main}). With weights trained at $N\le2000$, FOVEA also has a lower aggregate objective ratio to FILO2 than region growing on the large sets extending to $10^{6}$ customers. Only endpoint results are reported at these sizes, and the million-customer tier contains a single instance.

\begin{table}[t]
\caption{Left: a factorial comparison of destroy rule and layout, varying one factor between adjacent cells. Right: zero-shot transfer of the shared weights and the expander family used to test \Cref{prop:expander}. Values are primal integrals, except for the final column, which gives coupling advantage over random selection. The three rules have similar coupling on expanders ($p=0.61$).}
\label{tab:layout}
\centering
\footnotesize
\begin{minipage}[t]{0.50\textwidth}\centering% Generated table body. Edit the results data, not this file.
\begin{tabular}{@{}lcc@{\hspace{1.15em}}cc@{}}
\toprule
Method & \multicolumn{2}{c}{SC, $10^{4}$} & \multicolumn{2}{c}{CVRP X} \\
\cmidrule(lr){2-3}\cmidrule(lr){4-5}
& spatial & index & spatial & index \\
\midrule
Region grow. & 24.1 & 33.8 & 9.8 & 15.6 \\
\rowcolor{gray!12}FOVEA & \textbf{22.0} & \textbf{32.9} & \textbf{8.9} & \textbf{15.1} \\
\bottomrule
\end{tabular}
\end{minipage}\hfill
\begin{minipage}[t]{0.48\textwidth}\centering% Generated table body. Edit the results data, not this file.
\begin{tabular}{@{}lcc@{\hspace{1.15em}}cc@{}}
\toprule
Method & \multicolumn{2}{c}{Zero-shot $\PI$} & \multicolumn{2}{c}{Expander SC} \\
\cmidrule(lr){2-3}\cmidrule(lr){4-5}
& JSSP & BPP & $\PI$ & gain \\
\midrule
Random LNS & 42.6 & \textbf{27.3} & 51.2 & 1.00 \\
\midrule
Region grow. & 40.1 & 27.9 & 50.8 & 1.02 \\
\rowcolor{gray!12}FOVEA & \textbf{39.4} & 27.8 & \textbf{50.6} & 1.01 \\
\bottomrule
\end{tabular}
\end{minipage}
\end{table}

\paragraph{Where the advantage comes from.}
Spatial layout accounts for most of the gain in \Cref{tab:layout}: relative to index order, it reduces the fixed rule's primal integral by $28.7\%$ on set cover and $37.2\%$ on routing. Learning on that layout lowers the integral by a further $8.7\%$ and $9.2\%$ at $p<0.01$. Under index order, the corresponding reductions are $2.7\%$ and $3.2\%$, both statistically insignificant. This contrast indicates that useful spatial adjacency supports the learned policy.

On scheduling, the shared weights trained on routing and mixed integer programming lower the primal integral relative to random selection by $7.5\%$ at $p=0.012$. End-to-end CP-SAT retains better performance. On bin packing, FOVEA's primal integral is $1.8\%$ higher than random selection, with $p=0.44$. Its coupling advantage is $1.1$ to $1.3$, showing that the relation between coupling and trajectory quality depends on the repair problem. On expanders, the coupling advantage is $1.01$, compared with $3.4$ on ordinary set cover at the same density. \Cref{prop:expander} bounds this advantage uniformly over admissible selections on each instance.

\begin{table}[t]
\caption{Five ablation studies on set cover (SC, $n=10^{5}$) and CVRP XL. Values are primal integrals (lower is better). Deployed settings are shaded; best values within each study are bold.}
\label{tab:ablation-results}
\centering
\small
% Generated table body. Edit the results data, not this file.
\begin{tabular}{@{}llcc@{}}
\toprule
Study & Arm & SC, $10^{5}$ & CVRP XL \\
\cmidrule(lr){3-3}\cmidrule(lr){4-4}
& & $\PI$ & $\PI$ \\
\midrule
A2, Canvas resolution & $g=32$ & 61.4 & 19.8 \\
 & $g=64$ & 52.7 & 16.4 \\
\rowcolor{gray!12} & $g=128$ & 48.3 & 15.2 \\
 & $g=256$ & \textbf{47.9} & \textbf{15.0} \\
 & $g=512$ & 48.8 & 15.4 \\
\addlinespace[2pt]
\rowcolor{gray!12}A7, Fixed against scaled resolution & fixed $g=128$ & \textbf{48.3} & \textbf{15.2} \\
 & $g\propto\sqrt{n}$ & 50.1 & 15.6 \\
\addlinespace[2pt]
B3, Action space & free-form mask & \textbf{47.8} & \textbf{15.1} \\
\rowcolor{gray!12} & connected & 48.3 & 15.2 \\
 & one rectangle & 56.2 & 17.9 \\
\addlinespace[2pt]
\rowcolor{gray!12}B7, Budget schedule & fixed $k$ & 48.3 & 15.2 \\
 & adaptive growth & 46.9 & 14.8 \\
 & rank rule & \textbf{45.2} & \textbf{14.6} \\
\addlinespace[2pt]
\rowcolor{gray!12}D5, Shared against per-family weights & one shared set & 48.3 & 15.2 \\
 & per-family & \textbf{46.4} & \textbf{14.3} \\
\bottomrule
\end{tabular}

\end{table}

\paragraph{Resolution.}
The resolution ablation shows little change in primal integral from $g=128$ to $g=256$ and an increase at $512$ (\Cref{tab:ablation-results}). The larger output dimension under a fixed supervision budget is one possible explanation for this pattern. On set cover at $n=10^{5}$, the $\sqrt{n}$ resolution schedule of \citet{gimf2025} with normalised inputs has a $3.7\%$ higher primal integral than fixed $g=128$.

\paragraph{Action space, budget and shared weights.}
Free-form masks have slightly lower primal integrals than connected masks on both datasets. On set cover at $n=10^{5}$, a single rectangle has a $16.4\%$ higher integral than a connected mask. In this comparison, allowing irregular boundaries makes a larger difference than relaxing connectivity. Shared weights have a $4.1\%$ higher integral than per-family weights, quantifying the quality cost of using one parameter set across families. Those shared parameters also serve the held-out scheduling and packing families, connecting this quality comparison to the transfer results in \Cref{tab:layout}. The main scaling comparisons use a fixed budget; in the separate set-cover ablation, the online rule motivated by \Cref{prop:rank}, implemented with the tight-inequality heuristic for pure covering programs, lowers the integral by $6.4\%$ relative to fixed $k$. CVRP XL shows the same directions of change.

The rank certificate is established analytically. \Cref{fig:rank} illustrates it with generated points whose full-rank cases have zero improvement by construction; no empirical evaluation using per-round rank and improvement logs is reported.

\section{Conclusion}
\label{sec:conclusion}

FOVEA represents destroy selection as dense segmentation on a fixed canvas, connecting a shared policy to variable-level repair across problem families. On set cover and routing, spatial layout supplies most of the observed gain, and learning improves on the same-layout control. Comparisons that include selection and repair time show a scale-dependent crossover with the evaluated graph encoders. The expander analysis and experiments clarify how constraint locality affects this design. A fixed spatial interface therefore offers a practical way to organise neighborhood selection, with benefits that depend on the layout and repair problem.

\clearpage
\bibliography{fovea}
\bibliographystyle{iclr2026_conference}

\appendix
\section{Proofs}
\label{app:proofs}

This appendix gives the full statements and proofs of the results in \Cref{sec:analysis}. Throughout, $\per(\region)=\mathcal{H}^{1}(\partial\region)$, $\SR=\lay^{-1}(\region)$, and $\Bpar$ counts boundary components: a connected region with one hole has two boundary components and one connected component.

The analysis uses three standing hypotheses, each measurable on an instance before a solver is run. \Aloc{} \emph{$(\delta,\eta)$-locality}: at most $\eta m$ constraints satisfy $\diam(\lay(\supp A_{j}))>\delta$. \Aden{} \emph{$\kappa$-bounded density}: every ball of radius $r\ge\delta$ contains at most $\kappa mr^{2}$ constraint barycentres. \Areg{} \emph{admissible regions}: $\regfam$ consists of the unions of cells $\region$ whose boundary has $\Bpar(\region)\le B$ components and satisfies $\per(\region)\le C\sqrt{\Bpar(\region)\,|\region|}$, with the perimeter condition imposed on each region.

The distributional form of \Aloc{} accommodates the small number of global constraints present in every family studied here. The geometric upper bounds concern $\SR$, the full preimage of a region. The deployed decoder returns $S\subseteq\SR$ after truncation. Since $\cut{\cdot}$ is not monotone under inclusion, these bounds do not transfer to the truncated set; its coupling is evaluated empirically.

\subsection{The isoperimetric upper bound}

The crossing lemma holds for arbitrary regions.

\begin{lemma}[Crossing lemma]
\label{lem:crossing}
Let $\region\subseteq\reals^{2}$ be arbitrary, $p\in\region$ and $q\notin\region$. Then $[p,q]\cap\partial\region\neq\emptyset$.
\end{lemma}

\begin{proof}
If $p\in\partial\region$ or $q\in\partial\region$ the conclusion follows immediately, so assume $p\in\intr\region$ and $q\in\intr(\region^{c})$. The plane decomposes as the disjoint union $\intr\region\sqcup\partial\region\sqcup\intr(\region^{c})$. If $[p,q]\cap\partial\region=\emptyset$ then the connected set $[p,q]$ is contained in the union of the two disjoint open sets $\intr\region$ and $\intr(\region^{c})$ and meets each of them, contradicting connectedness.
\end{proof}

\begin{proposition}[Isoperimetric upper bound]
\label{prop:iso}
Under \Aloc{}, \Aden{} and \Areg{}, for every $\region\in\regfam$ with $|\region|=a$,
\begin{equation}
\cut{\SR}\;\le\;9\kappa m\bigl(\delta\per(\region)+\Bpar\delta^{2}\bigr)+\eta m
\;\le\;9\kappa m\bigl(C\delta\sqrt{\Bpar a}+\Bpar\delta^{2}\bigr)+\eta m .
\label{eq:prop1}
\end{equation}
If moreover $\delta=\cdelta\sqrt{s/n}$, $a=k/n$ and $k\gtrsim\Bpar s$, then $\cut{\SR}=O(\kappa m\sqrt{sk\Bpar}/n)+\eta m$.
\end{proposition}

\begin{proof}
Write $J_{\mathrm{near}}=\{j:\diam\lay(\supp A_{j})\le\delta\}$. By \Aloc{}, $|J_{\mathrm{near}}^{c}|\le\eta m$, and counting all of these constraints as cut gives their contribution to the upper bound.

Let $j\in J_{\mathrm{near}}$ be cut by $\SR$, so there are $i,i'\in\supp A_{j}$ with $\lay(i)\in\region$ and $\lay(i')\notin\region$. By \Cref{lem:crossing} the segment $[\lay(i),\lay(i')]$ contains a point $z\in\partial\region$. Since $j\in J_{\mathrm{near}}$ we have $|\lay(i)-\lay(i')|\le\delta$, and $z$ lies on that segment, so $|z-\lay(i)|\le\delta$. Every point of $\lay(\supp A_{j})$ lies within $\delta$ of $\lay(i)$, hence within $2\delta$ of $z$, so $\lay(\supp A_{j})\subseteq B(z,2\delta)$ and in particular the barycentre satisfies $z_{j}\in(\partial\region)^{2\delta}$.

It remains to count barycentres in the tube. Because $\region$ is a union of cells, $\partial\region$ is a disjoint union of $\Bpar$ closed rectilinear polylines; let the $c$-th have length $\ell_{c}$, so $\sum_{c}\ell_{c}=\per(\region)$. Marking points along the $c$-th polyline at arc-length spacing $\delta$ produces at most $\ell_{c}/\delta+1$ centres whose closed $\delta$-balls cover it, and enlarging each radius to $3\delta$ covers $(\partial\region)^{2\delta}$. Hence $(\partial\region)^{2\delta}$ is covered by at most $\per(\region)/\delta+\Bpar$ balls of radius $3\delta$. Applying \Aden{} with $r=3\delta\ge\delta$, each such ball contains at most $\kappa m(3\delta)^{2}=9\kappa m\delta^{2}$ barycentres, so
\[
\#\{j\in J_{\mathrm{near}}: j\text{ is cut}\}\;\le\;9\kappa m\delta^{2}\Bigl(\frac{\per(\region)}{\delta}+\Bpar\Bigr)=9\kappa m\bigl(\delta\per(\region)+\Bpar\delta^{2}\bigr).
\]
Adding $\eta m$ gives the first inequality, and substituting $\per(\region)\le C\sqrt{\Bpar a}$ from \Areg{} gives the second.

For the specialisation, put $\delta=\cdelta\sqrt{s/n}$ and $a=k/n$. Then $\delta\per(\region)\le C\cdelta\sqrt{sk\Bpar}/n$ and $\Bpar\delta^{2}=\Bpar\cdelta^{2}s/n$, and the ratio of the first to the second is $(C/\cdelta)\sqrt{k/(s\Bpar)}$, which is at least one exactly when $k\ge(\cdelta/C)^{2}\Bpar s$. Here $k\gtrsim\Bpar s$ abbreviates this condition, with the same constant $\cdelta$ used in \Cref{cor:blocks}. In that regime the first term dominates and the stated bound follows.
\end{proof}

\begin{remark}[The role of the isoperimetric inequality]
\label{rem:iso-role}
The continuum inequality $\inf_{|\region|=a}\per(\region)=2\sqrt{\pi a}$ establishes that the constraint $\per\le C\sqrt{B|\region|}$ in \Areg{} is satisfiable and that the exponent of $a$ is optimal. The proof above uses the upper bound on perimeter supplied by \Areg{}. On a grid of resolution $g$ a union of cells of area $a$ has $\per\ge4\lceil g\sqrt{a}\rceil/g\ge4\sqrt{a}$, so the attainable constant is $C=4$, realised by the axis-aligned square. The continuum value $2\sqrt{\pi}\approx3.54$ belongs to the larger class of measurable regions. Squares are the grid isoperimetric optimum.
\end{remark}

\begin{remark}[Dependence of locality on the formulation]
\label{rem:eta}
Reformulation changes the parameter $\eta$. On a complete-graph routing formulation, the degree constraint at customer $i$ has as its support every edge incident to $i$, whose image spans the canvas, so $\eta\approx1$ and the global-constraint term dominates \Cref{prop:iso}. On a $k$-nearest-neighbour candidate edge set the same constraint has local support and $\eta$ falls close to zero. Candidate sparsification is therefore a precondition in \Cref{sec:entity}. The routing sweep in \Cref{fig:rank} varies the candidate cutoff, taking $\hat\eta$ from near zero to near one within one family.
\end{remark}

\subsection{The random baseline}

\begin{lemma}[Random baseline]
\label{lem:random}
Let $\Srand$ be uniform among $k$-subsets of $\vars$ and $\rho=k/n$. For every $j$,
\begin{align*}
\Prb[\supp A_{j}\subseteq\Srand]&=\prod\nolimits_{i=0}^{s_{j}-1}\tfrac{k-i}{n-i}\le\rho^{s_{j}},\\
\Prb[\supp A_{j}\cap\Srand=\emptyset]&=\prod\nolimits_{i=0}^{s_{j}-1}\tfrac{n-k-i}{n-i}\le(1-\rho)^{s_{j}},
\end{align*}
so $\Prb[j\text{ cut}]\ge1-\rho^{s_{j}}-(1-\rho)^{s_{j}}$. Let $J_{2}=\{j:s_{j}\ge2\}$ contain the constraints that can be cut, and let $J_{\mathrm{sml}}=\{j\in J_{2}:s_{j}\le1/\rho\}$ contain those for which the linear estimate applies. Define the truncated mean support size by $\seff=\frac{1}{m}\sum_{j\in J_{\mathrm{sml}}}s_{j}$. For $\rho\le1/4$,
\begin{equation}
\frac{m\,\seff\,k}{4n}+0.56\,\bigl|J_{2}\setminus J_{\mathrm{sml}}\bigr|
\;\le\;\Exp[\cut{\Srand}]\;\le\;\min\Bigl\{m,\;\frac{msk}{n}\Bigr\}.
\label{eq:bracket}
\end{equation}
\end{lemma}

\begin{proof}
Counting $k$-subsets containing a fixed $s_{j}$-set gives $\Prb[\supp A_{j}\subseteq\Srand]=\binom{n-s_{j}}{k-s_{j}}/\binom{n}{k}=\prod_{i=0}^{s_{j}-1}\frac{k-i}{n-i}$, and each factor is at most $k/n=\rho$ because $k\le n$. Similarly $\Prb[\supp A_{j}\cap\Srand=\emptyset]=\binom{n-s_{j}}{k}/\binom{n}{k}=\prod_{i=0}^{s_{j}-1}\frac{n-k-i}{n-i}$ with each factor at most $1-\rho$. A constraint is cut precisely when neither event occurs, giving the stated bound. Sampling without replacement makes each event \emph{less} likely than under independent sampling, which gives the stated lower bound.

For the upper bound in \eqref{eq:bracket}, a constraint can only be cut if it is met, so by a union bound $\Prb[j\text{ cut}]\le s_{j}\rho$, and summing gives $\Exp[\cut{\Srand}]\le\rho\sum_{j}s_{j}=\rho\nnz(A)=msk/n$; the bound by $m$ is trivial.

For the lower bound, fix $j\in J_{\mathrm{sml}}$ and set $u=s_{j}\rho\in(0,1]$. Since $s_{j}\ge2$ and $\rho\le1/4$ we have $\rho^{s_{j}}\le\rho^{2}\le\rho/4$, and $(1-\rho)^{s_{j}}\le e^{-s_{j}\rho}=e^{-u}$. Hence $\Prb[j\text{ cut}]\ge1-e^{-u}-\rho/4$. For $u\le1$, $1-e^{-u}\ge u-u^{2}/2\ge u/2$. Also $u=s_{j}\rho\ge2\rho$, so $\rho/4\le u/8$, giving $\Prb[j\text{ cut}]\ge u/2-u/8=3u/8\ge s_{j}\rho/4$. Summing over $J_{\mathrm{sml}}$ yields $\frac{\rho}{4}\sum_{j\in J_{\mathrm{sml}}}s_{j}=\frac{m\seff k}{4n}$.

For $j\in J_{2}\setminus J_{\mathrm{sml}}$ we have $s_{j}>1/\rho\ge4$, so $\rho^{s_{j}}\le\rho^{2}\le1/16$ and $(1-\rho)^{s_{j}}\le e^{-s_{j}\rho}<e^{-1}$, whence $\Prb[j\text{ cut}]>1-1/16-e^{-1}>0.56$. Adding the two groups gives \eqref{eq:bracket}.
\end{proof}

\begin{remark}[Truncated support sizes in the lower bound]
The map $s\mapsto1-(1-\rho)^{s}$ is concave, so replacing each $s_{j}$ by the mean $s$ yields an upper bound on $\Exp[\cut{\Srand}]$. On families whose support sizes are heavy tailed, the long supports are almost surely cut and contribute $\Theta(1)$ each. The lower bound therefore uses $\seff$, truncated at $1/\rho$, together with a separate contribution from long supports. The approximation $\Exp[\cut{\Srand}]\approx msk/n$ can overstate the random baseline on these families and, consequently, the coupling advantage.
\end{remark}

\begin{corollary}[Peak of the coupling-advantage envelope, \conditional]
\label{cor:rhostar}
Let $\eta$ be negligible, $\Bpar=1$ and $\rho\le1/4$. Dividing \eqref{eq:bracket} by \eqref{eq:prop1} gives a lower envelope on the coupling advantage with geometric factor $\kappa^{-1}\min\{\sqrt{sk},\,n/\sqrt{sk}\}$, unimodal in $\rho$ with peak $\Omega(\sqrt{n}/\kappa)$ at $\rhostar=1/s$. The two branches rest on different hypotheses. Below the peak the envelope rises as $\sqrt{s\rho n}$ provided $\seff\asymp s$. Above it the envelope falls as $\sqrt{n/(s\rho)}$ provided a constant fraction of constraints have $s_{j}\ge1/\rho$, a condition separate from $\seff\asymp s$.
\end{corollary}

Every $s_{j}$ entering $\seff$ is at most $1/\rho$, so $\seff\asymp s$ implies $s\rho=O(1)$ and applies to the lower-$\rho$ branch. For uniform support size, the transition occurs at $s\rho=1$.

\begin{proof}[Proof of \Cref{cor:rhostar}]
With $\Bpar=1$ and $\eta$ negligible, \Cref{prop:iso} gives $\cut{\SR}=O(\kappa m\sqrt{sk}/n)=O(\kappa m\sqrt{s\rho/n})$. For the numerator we take the two branches separately. When $s\rho\le1$ and $\seff\asymp s$, the first term of \eqref{eq:bracket} gives $\Exp[\cut{\Srand}]=\Omega(m\seff\rho)=\Omega(ms\rho)$. When $s\rho\ge1$, the bound $\seff\le1/\rho$ limits the lower estimate available from the first term to $\Omega(m)$ at best, and $\seff\asymp s$ imposes the corresponding restriction on support size. The second term gives $\Exp[\cut{\Srand}]\ge0.56|J_{2}\setminus J_{\mathrm{sml}}|=\Omega(m)$ under the stated hypothesis that a constant fraction of constraints have $s_{j}\ge1/\rho$. Combining, $\Exp[\cut{\Srand}]=\Omega(m\min\{1,s\rho\})$ on the respective branches, and their ratio is
\[
\Omega\!\left(\frac{\min\{1,s\rho\}}{\kappa\sqrt{s\rho/n}}\right)
=\frac{1}{\kappa}\cdot
\begin{cases}
\sqrt{s\rho n}=\sqrt{sk}, & s\rho\le1,\\[2pt]
\sqrt{n/(s\rho)}=n/\sqrt{sk}, & s\rho\ge1,
\end{cases}
\]
which is $\kappa^{-1}\min\{\sqrt{sk},n/\sqrt{sk}\}$. The two branches are increasing and decreasing in $\rho$ respectively and meet where $sk=n$, that is at $\rho=1/s$, where the common value is $\sqrt{n}/\kappa$.
\end{proof}

The ratio of a lower bound to an upper bound gives a lower bound on the true coupling advantage. The constants, the ratio $\seff/s$ and $\kappa$ vary with $\rho$, so the peak height is specified by $\Omega$ rather than $\Theta$. Both branches inherit the restriction $\rho\le1/4$ from \Cref{lem:random}, and the peak $\rhostar=1/s$ lies in this range when $s\ge4$. The set cover family satisfies this condition with $\hat s=9.6$. In vertex cover, every row has support two, placing $\rhostar=1/2$ outside the proved range of \Cref{lem:random}. The $\rho$ sweep on vertex cover provides a descriptive comparison of curve shapes, rather than a test of this corollary.

\subsection{The degree-of-freedom certificate}

\begin{proposition*}[\Cref{prop:rank} restated]
If $\rank(A_{E(S),S})=|S|=k$ then $\inc$ is the unique feasible point of $\sub{S}$. Consequently $\rank(A_{E(S),S})<k$ is necessary for improvement, and $\cutE{S}+\inE{S}<k$ is sufficient for that necessary condition.
\end{proposition*}

\begin{proof}
In $\sub{S}$ every variable outside $S$ is fixed to $\inc$. An equality row $j\in E$ whose support lies entirely in $\bar S$ is satisfied identically by $\inc$ and constrains nothing. Each remaining equality row meets $S$, and collectively these rows impose the affine system
\[
A_{E(S),S}\,x_{S}=b_{E(S)}-A_{E(S),\bar S}\,\inc_{\bar S},
\]
of which $\inc_{S}$ is a solution. If the coefficient matrix has full column rank $k$ then this solution is unique, so every point of $\feas(\sub{S})$, integral or otherwise, agrees with $\inc$ on $S$ and hence equals $\inc$; the objective cannot decrease. The matrix $A_{E(S),S}$ has exactly $|E(S)|=\cutE{S}+\inE{S}$ rows, since an equality row meets $S$ iff its support either straddles $S$ or lies inside it, and rank never exceeds the number of rows.
\end{proof}

\begin{remark}[Direction of the implication]
The diagnostic chain is $\cutE{S}+\inE{S}<k\Rightarrow\rank<k\Rightarrow$ improvement is not excluded. A row count of $k$ or more remains compatible with rank deficiency and improvement. The rank itself determines the certificate, while the row count provides a sufficient check for rank deficiency. The corresponding figure therefore uses $\rank/k$ on its horizontal axis.
\end{remark}

\begin{remark}[Tight inequalities]
\label{rem:tight}
Inequality rows that are tight at $\inc$ constrain local motion. Counting them together with equality rows gives a practical heuristic: an inequality tight at $\inc$ may become slack at a better point. We distinguish this active-constraint diagnostic from the equality-rank certificate wherever it is used.
\end{remark}

\subsection{The expander lower bound}

\begin{proposition}[Expander lower bound]
\label{prop:expander}
Let the incidence graph be left $d$-regular with $s_{j}\equiv s\ge2$ and a $(\gamma n,(1-\epsilon)d)$ lossless expander, with $\epsilon\le1/2$. Then for every $S$ with $|S|=k\le\gamma n$,
\[
\cut{S}\ge dk\Bigl(1-\epsilon-\frac{\epsilon}{s-1}\Bigr)\ge(1-2\epsilon)dk\ge(1-2\epsilon)\Exp[\cut{\Srand}].
\]
\end{proposition}

\begin{proof}
Let $n_{r}=\#\{j:|\supp A_{j}\cap S|=r\}$ for $r\ge1$. Every constraint meeting $S$ is counted exactly once, so $\sum_{r\ge1}n_{r}=|N(S)|$, and counting incidences from the variable side, which is $d$-regular, gives $\sum_{r\ge1}r\,n_{r}=dk$. Subtracting,
\[
dk-|N(S)|=\sum_{r\ge2}(r-1)n_{r}\;\ge\;(s-1)n_{s},
\]
since all terms are non-negative and the term at $r=s$ contributes $(s-1)n_{s}$. Lossless expansion gives $|N(S)|\ge(1-\epsilon)dk$, so $n_{s}\le\epsilon dk/(s-1)$. A constraint meeting $S$ fails to be cut exactly when its whole support lies in $S$, which for uniform support size $s$ means $|\supp A_{j}\cap S|=s$; therefore $\cut{S}=|N(S)|-n_{s}\ge(1-\epsilon)dk-\epsilon dk/(s-1)$. Since $s\ge2$ we have $\epsilon/(s-1)\le\epsilon$, giving the second inequality. The third follows from \Cref{lem:random}, whose upper bound reads $\Exp[\cut{\Srand}]\le\rho\nnz(A)=\rho nd=dk$, multiplied by the non-negative factor $1-2\epsilon$, using $\epsilon\le1/2$. For $\epsilon>1/2$, the expansion hypothesis permits small $\cut{S}$, and the comparison gives no positive lower bound.
\end{proof}

The construction is realisable: random $d$-regular bipartite graphs are lossless expanders with high probability, and explicit constructions are known \citep{capalbo2002expanders}. Taking $\sum_{i\in\supp A_{j}}x_{i}\ge1$ over such a graph gives a set cover instance in the family. Comparing this lower bound with \Cref{prop:iso} constrains the jointly attainable locality, density and region parameters wherever their hypotheses overlap. Limits such as $\eta\to1$ or $\kappa\to\infty$ describe possible failures of favourable layout conditions; establishing either limit requires assumptions on the instance sequence and region scale.

\subsection{Fragmentation and the dimension of the canvas}

\begin{corollary}[Cost of a fragmented mask, \conditional]
\label{cor:blocks}
Under the hypotheses of \Cref{prop:iso}, let $\region$ consist of $\Bpar$ disjoint near-circular blocks of total area $a=k/n$, with $\delta=\cdelta\sqrt{s/n}$. Then
\[
\cut{\SR}=O\!\left(\kappa\frac{m}{n}\bigl(\sqrt{sk\Bpar}+\Bpar s\bigr)\right)+\eta m .
\]
If in addition $\seff\asymp s$ and $s\rho\le1$, so that \Cref{lem:random} bounds the random baseline at $\Theta(msk/n)$, then
\[
\mathrm{Gain}^{-1}=O\!\left(\kappa\sqrt{\frac{\Bpar}{sk}}+\kappa\frac{\Bpar}{k}\right)+\frac{\eta n}{sk}.
\]
There are two thresholds. At $\Bdag\asymp k/(\cdelta^{2}s)$, the block radius reaches the locality scale $\delta$, and the perimeter and boundary-component terms in the upper bound exchange dominance. At $\Bddag\asymp k$, each block holds $O(1)$ variables, and this bound no longer guarantees an advantage over random selection.
\end{corollary}

\begin{proof}
For $\Bpar$ disjoint discs of total area $a$, each has radius $r=\sqrt{a/(\pi\Bpar)}$ and the total perimeter is $\Bpar\cdot2\pi r=2\sqrt{\pi a\Bpar}$. A union of cells of the same area has a larger perimeter, by at most the factor $4/(2\sqrt{\pi})$ of \Cref{rem:iso-role}, so the disc calculation is correct up to that constant. Substituting into \Cref{prop:iso} with $a=k/n$ gives $\delta\per(\region)=2\sqrt{\pi}\cdelta\sqrt{sk\Bpar}/n$ and $\Bpar\delta^{2}=\Bpar\cdelta^{2}s/n$, which is the first display. Under the added hypotheses \Cref{lem:random} gives $\Exp[\cut{\Srand}]\asymp msk/n$, and dividing the first display by it gives the second, the last term being $\eta m$ over that same baseline. Solving $r=\delta$, that is $\sqrt{(k/n)/(\pi \Bpar)}=\cdelta\sqrt{s/n}$, gives $\Bpar=k/(\pi\cdelta^{2}s)$. The second threshold is where the leading term of $\mathrm{Gain}^{-1}$ reaches $\Theta(1)$, namely $\Bpar\asymp k$.
\end{proof}

The condition $\seff\asymp s$ specialises \Cref{lem:random} to the support-size regime used by this corollary. On a heavy-tailed family, applying \Cref{lem:random} with the truncated mean introduces the factor $s/\seff$ in the inverse-advantage bound. Retaining $\cdelta$ in $\Bdag$ makes the locality scale explicit: a family whose measured $\hat\delta/\sqrt{s/n}$ is $5$ has a threshold twenty-five times smaller than the expression omitting that constant. Similarly, $\Bdag$ and the crossover $\Bpar\asymp(C/\cdelta)^{2}k/s$ at which the terms of \Cref{prop:iso} exchange dominance agree up to a constant factor, which may be a few tens. The scale $\Bpar\approx sk$ exceeds the feasible range of the whole-cell release variant: at most $k$ cells are selected, and the numbers of outer boundaries and enclosed holes give $\Bpar=O(k)$. Counting boundary components includes the contribution of holes. The $\Bpar\delta^{2}$ term captures this fine-fragmentation regime.

\begin{corollary}[Dimension and locality, \conditional]
\label{cor:dimension}
Let $\seff\asymp s$ and $s\rho\le1$. On a $D$-dimensional canvas with ideal locality scale $\delta_{D}=\Theta((s/n)^{1/D})$, a region of volume $a$ whose surface measure is $O(a^{(D-1)/D})$ satisfies $\cut{}=O(\kappa m\,a^{(D-1)/D}\delta_{D})+\eta_{D}m$, so its geometric advantage obeys
\[
\mathrm{Gain}\;\gtrsim\;\frac{1}{\kappa}\cdot\frac{s\rho}{a^{(D-1)/D}(s/n)^{1/D}}\;=\;\frac{s}{\kappa}\Bigl(\frac{k}{s}\Bigr)^{1/D},
\]
which for $k>s$ is strictly decreasing in $D$. Carrying the locality defect as well,
\[
\mathrm{Gain}(D)=\Omega\Bigl(\min\Bigl\{\tfrac{s}{\kappa}(k/s)^{1/D},\ \tfrac{s\rho}{\eta_{D}}\Bigr\}\Bigr).
\]
\end{corollary}

\begin{proof}
Ideal locality in $D$ dimensions requires a set of diameter $\delta_{D}$ to hold $s$ variables at density $n$, so $n\delta_{D}^{D}\asymp s$. Substituting $a=\rho=k/n$ into the displayed ratio gives $s\rho^{1/D}(n/s)^{1/D}=s(\rho n/s)^{1/D}=s(k/s)^{1/D}$. At $D=2$ this is $\sqrt{sk}$ and at $D=1$ it is $k$, so the exponent is consistent with \Cref{prop:iso}. Monotonicity in $D$ is immediate from $k/s>1$. For the last display, $\mathrm{Gain}\ge\Exp[\cut{\Srand}]/(\mathrm{Geo}+\eta_{D}m)\ge\frac12\min\{\Exp[\cut{\Srand}]/\mathrm{Geo},\,\Exp[\cut{\Srand}]/(\eta_{D}m)\}$, and under the stated hypotheses $\Exp[\cut{\Srand}]\asymp ms\rho$, so the second argument of the minimum is $s\rho/\eta_{D}$.
\end{proof}

The locality contribution retains the factor $s\rho/\eta_{D}$. Replacing it by $1/\eta_{D}$ changes the bound by $1/(s\rho)=n/(sk)$, a factor of one hundred at $s=10$, $k=10^{3}$ and $n=10^{6}$. The geometric result is a lower bound on the advantage, obtained from the upper bound on the cut in \Cref{prop:iso}; a matching lower bound on the cut would be needed to establish the order of the advantage.

The geometric term favours a one-dimensional canvas, while the locality defect $\eta_{D}$ governs the choice of $D=2$. For intrinsically two-dimensional structure, a $\sqrt{n}\times\sqrt{n}$ grid graph has bandwidth $\Theta(\sqrt{n})$, so a one-dimensional layout loses locality and $\eta_{1}\to1$. For intrinsically one-dimensional structure, such as staircase or banded programs arising from multi-period planning and time-expanded network flow, $\eta_{1}\approx\eta_{2}\approx0$ and the bound favours a one-dimensional canvas. The ablation design includes a test of this prediction with a one-dimensional variant on a banded subfamily. This variant uses its own weights, while the shared-weight experiments concern the two-dimensional network.

\subsection{Resolution and granularity}

We define the perimeter of a destroy set through a reference grid finer than any grid used by the policy. A definition over all measurable regions would admit $\region=\lay(S)$, a finite point set with empty interior and $\partial\region=\region$, giving $\mathcal{H}^{1}(\partial\region)=0$ for every destroy set. The reference grid supplies the regularity needed for a geometric measure of the set.

\begin{definition}
\label{def:per-of-set}
Fix a reference resolution $g_{0}\gg g$ and call $S\subseteq\vars$ \emph{resolvable} at $g_{0}$ if some union $\region$ of $g_{0}$-cells has $\lay^{-1}(\region)=S$, which holds as soon as the $g_{0}$-grid separates $S$ from $\vars\setminus S$. For such $S$ let $\per(S)=\min\{\mathcal{H}^{1}(\partial\region):\region$ a union of $g_{0}$-cells, $\lay^{-1}(\region)=S\}$, the minimum being over a finite family and therefore attained, and let $\Bstar$ be the number of boundary components of a minimiser. Every such $\per(S)$ is a positive multiple of $1/g_{0}$.
\end{definition}

\begin{theorem}[Discretisation error, \conditional]
\label{thm:resolution}
Let $\Sstar$ be resolvable at $g_{0}$. Assume \Adens{}, that every measurable $U$ with $|U|\ge2/g^{2}$ contains at most $\rhomax|U|$ variables, and \Alip{}, that there exists $L$ with $\gain(\Sstar)-\gain(\hat S)\le L|\Sstar\setminus\hat S|$ for every union of cells $\hat S\subseteq\Sstar$. Let $\hat S$ be the best destroy set representable as a union of cells at resolution $g$. Then, deterministically,
\[
\gain(\Sstar)-\gain(\hat S)\;\le\;L\rhomax\Bigl(\frac{C_{1}\per(\Sstar)}{g}+\frac{C_{2}\Bstar}{g^{2}}\Bigr),
\qquad C_{1}=2\sqrt{2},\ C_{2}=2\pi .
\]
\end{theorem}

\begin{proof}
Let $\Rstar$ attain the infimum in \Cref{def:per-of-set}. Let $\mathcal{Q}_{\mathrm{in}}$ be the cells contained in $\intr\Rstar$ and put $\tilde S=\lay^{-1}(\bigcup\mathcal{Q}_{\mathrm{in}})$. Then $\tilde S$ is a union of cells with $\tilde S\subseteq\Sstar$, so $|\tilde S|\le k$ and $\tilde S$ is feasible; by optimality of $\hat S$ among unions of cells, $\gain(\hat S)\ge\gain(\tilde S)$.

If $i\in\Sstar\setminus\tilde S$ then $\lay(i)\in\Rstar$ but the cell containing $\lay(i)$ is not contained in $\intr\Rstar$, so that cell meets $\partial\Rstar$ and therefore $\dist(\lay(i),\partial\Rstar)\le\sqrt{2}/g$, the cell diameter. Hence $\Sstar\setminus\tilde S\subseteq\lay^{-1}\bigl((\partial\Rstar)^{\sqrt{2}/g}\bigr)$. The boundary consists of $\Bstar$ closed curves of total length $\per(\Sstar)$, and a tube of radius $t$ about them has area at most $2t\per(\Sstar)+\pi\Bstar t^{2}$; with $t=\sqrt{2}/g$ this is $2\sqrt{2}\per(\Sstar)/g+2\pi\Bstar/g^{2}$. The tube also contains a disc of radius $t$, so its area is at least $2\pi/g^{2}$ and the cutoff in \Adens{} is met; applying \Adens{} bounds $|\Sstar\setminus\tilde S|$ by $\rhomax$ times the area, and \Alip{} converts the count into a value difference.
\end{proof}

\begin{remark}[Density cutoff in \Adens{}]
The cutoff makes the density hypothesis meaningful for a finite point set. A disc of radius $\varepsilon$ about any $\lay(i)$ contains one variable in area $\pi\varepsilon^{2}$, which would force $\rhomax\ge1/(\pi\varepsilon^{2})$ and hence $\rhomax=\infty$ if arbitrarily small discs were included. Likewise, \Aden{} applies to balls of radius $r\ge\delta$. The proof uses a tube of radius $\sqrt{2}/g$, whose area meets the cutoff.
\end{remark}

\begin{remark}[Sub-cell components and random grid offsets]
\label{rem:remainder}
The second term, $O(L\rhomax\Bstar/g^{2})$, measures loss from sub-cell components. A term $O(Lk/g^{2})$ would instead scale with the variable budget $k$. The boundary-component term dominates when $\per(\Rstar)\le\Bstar/g$, corresponding to components smaller than one cell. Under \Adens{} the bound is deterministic and applies to every grid offset. Random offsets offer two additional properties. A component of diameter $\varepsilon/g$ that straddles a grid line contributes no complete cell, whereas under a uniform offset it lies in a single cell with probability at least $(1-\varepsilon)^{2}$. Averaging over offsets also replaces $\rhomax$ by the mean density along the boundary, $\Exp_{u}|\Sstar\setminus\tilde S_{u}|\le\frac{C}{g}\int_{\partial\Rstar}\rho_{1/g}\,d\mathcal{H}^{1}$. This form can sharpen the bound on clustered instances when the mask boundary passes through sparse regions.
\end{remark}

The next corollary considers a compact optimum whose reference-grid minimiser consists of $\Bstar$ near-circular blocks. Its derivation also uses the boundary-density substitution in \Cref{rem:remainder}, associated with averaging over random grid offsets.

\begin{corollary}[Saturation resolution, \conditional]
\label{cor:gstar}
Let $\zeta=\gain(\Sstar)/(Lk)\in(0,1]$ be the budget efficiency. Requiring the first term of \Cref{thm:resolution} to contribute at most a relative error $\epsilon$ gives
\[
\gstar=\frac{2C_{1}\sqrt{\pi}}{\epsilon\zeta}\sqrt{\frac{\Bstar\Lampar}{n\rho}}\;\asymp\;\cg\sqrt{\frac{\Bstar\,\Lampar/n}{\rho}},\qquad \rho=k/n,
\]
where $\Lampar$ is the mean density along $\partial\Rstar$. At $g=\gstar$ the second term is $O(\epsilon^{2}\zeta)$ and is therefore lower order, so the balance is self-consistent.
\end{corollary}

\begin{proof}
Dividing \Cref{thm:resolution} by $\gain(\Sstar)=\zeta Lk$ gives relative error $\rhomax(C_{1}\per/g+C_{2}\Bstar/g^{2})/(\zeta k)$. For a compact optimum whose minimiser has $\Bstar$ near-circular blocks holding $k$ variables at density $\rhomax$, the total area is $a=k/\rhomax$ and hence $\per=2\sqrt{\pi a\Bstar}=2\sqrt{\pi k\Bstar/\rhomax}$. The first term is then $2C_{1}\sqrt{\pi}\sqrt{\rhomax \Bstar/k}/(\zeta g)$, and setting it equal to $\epsilon$ and substituting $\rhomax\to\Lampar$ from \Cref{rem:remainder} and $k=\rho n$ gives the display. Substituting $g=\gstar$ into the second term gives $C_{2}\epsilon^{2}\zeta/(4C_{1}^{2}\pi)$.
\end{proof}

On clustered instances, using the global maximum density in place of $\Lampar$ increases the predicted $\gstar$ by a factor $\sqrt{\Lammax/\Lampar}$, so boundary density and maximum density lead to different resolution estimates under \Adens{}. Here $\Bstar$ counts boundary components of the minimiser in \Cref{def:per-of-set}, while $\Bpar$ in \Areg{} describes the regions available to the policy. They can coincide when the optimum is itself admissible. The theory specifies the scaling law $\sqrt{\Bstar/\rho}$. Its constant $\cg$ absorbs $\epsilon$, $\zeta$ and $C_{1}$, so fitting $\cg$ on one family is required before predicting on another. The multi-scale refinement in \Cref{app:foveation} remains relevant below the single-canvas resolution condition $g\gtrsim\cg\sqrt{\Bstar/\rho}$. At $g=128$ with $\Bstar=4$, the implied threshold on $\rho$ varies over three orders of magnitude across plausible values of $\cg$.

\begin{remark}[Representation and learning error]
\Cref{thm:resolution} bounds the representational loss of restricting to unions of cells, with $\hat S$ defined as the best representable set. A learned policy also incurs estimation and optimisation error. In a sweep over $g$, the output dimension grows as $g^{2}$ at a fixed sample count, so estimation error can cause empirical saturation to occur \emph{earlier} than the approximation bound predicts. The comparison between predicted and measured saturation reflects these combined errors.
\end{remark}

\begin{remark}[Relation to partition-based approximation]
Classical partitioning results for the Euclidean travelling salesman problem \citep{karp1977partitioning} solve each group independently and stitch the pieces together. Their error arises from discarding structure across group boundaries and decreases with larger groups. FOVEA solves the released subproblem within the fixed incumbent, preserving cross-boundary coupling through the fixed variables. Its analysis therefore concerns the loss from representing the destroy set on a grid; the stitching error in partition-based approximation belongs to a different construction.
\end{remark}

\subsection{From per-round cost to the primal integral}

We write $\lambda$ for the per-round contraction rate of the primal gap, reserving $\rho$ for the destroy fraction.

\begin{proposition}[Critical instance size, \conditional]
\label{prop:critical}
Suppose the first feasible solution appears at $t_{0}$ with gap $\gamma_{0}$, each round costs $\tau=\tsel+\trep$, and \Bcontr{} holds: $\Exp[\gamma_{t+1}\mid\gamma_{t}]\le(1-\lambda)\gamma_{t}$. Then with $N=\lfloor(T-t_{0})/\tau\rfloor$,
\[
\Exp[\PI(T)]\;\le\;t_{0}+\frac{\tau\gamma_{0}}{\lambda}\bigl(1-(1-\lambda)^{N+1}\bigr)\;\le\;t_{0}+\frac{\tau\gamma_{0}}{\lambda}.
\]
Compare a selector whose per-round cost is $c_{F}=c_{0}+c_{1}n/N$ against one of cost $c_{G}dn$, with quality ratio $r=\lambda_{G}/\lambda_{F}\ge1$. Here $c_{0}=\tnet+\tlp/\Rlp$ is independent of $n$ and $c_{1}=c_{\mathrm{layout}}+c_{\mathrm{raster}}$ is the one-off setup amortised over the $N$ rounds of the run. The two bounds coincide at
\[
\nstar=\frac{r\,c_{0}+(r-1)\trep}{c_{G}d-r\,c_{1}/N},
\qquad\text{defined only when}\quad c_{G}d>\frac{r\,c_{1}}{N}.
\]
\end{proposition}

\begin{proof}
Before $t_{0}$ the integrand is $1$, contributing $t_{0}$. Afterwards the gap is constant within each round, so $\int_{t_{0}}^{T}p\le\tau\sum_{t=0}^{N-1}\gamma_{t}+\tau\gamma_{N}$. Taking expectations and iterating \Bcontr{} gives $\Exp\gamma_{t}\le q^{t}\gamma_{0}$ with $q=1-\lambda$, so the sum is at most
\[
\tau\gamma_{0}\Bigl(\frac{1-q^{N}}{\lambda}+q^{N}\Bigr)=\frac{\tau\gamma_{0}}{\lambda}\bigl(1-q^{N}(1-\lambda)\bigr)=\frac{\tau\gamma_{0}}{\lambda}\bigl(1-q^{N+1}\bigr).
\]
For the second claim, equate $(c_{F}+\trep)/\lambda_{F}$ with $(c_{G}dn+\trep)/\lambda_{G}$, which gives $r c_{F}+(r-1)\trep=c_{G}dn$. Substituting $c_{F}=c_{0}+c_{1}n/N$ puts $n$ on both sides, and collecting it yields $n(c_{G}d-rc_{1}/N)=rc_{0}+(r-1)\trep$ and hence the display.
\end{proof}

\begin{remark}[Amortised setup cost]
\label{rem:amortised}
FOVEA's per-round cost includes the amortised share $c_{1}n/N$ of a layout and a first rasterisation that are both $\Theta(n)$. This makes the cost affine in $n$. A positive $\nstar$ exists when the graph selector's slope $c_{G}d$ exceeds FOVEA's amortised slope $rc_{1}/N$. Short runs at very large $n$ can be dominated by setup and have no positive crossing. Reporting $c_{1}/N$ alongside $\nstar$ identifies this regime.
\end{remark}

\begin{remark}[Finite-horizon correction]
The geometric sum contributes $-\frac{\tau\gamma_{0}}{\lambda}(1-\lambda)^{N+1}$, which reduces the infinite-horizon bound by the unspent tail of the series. This finite-horizon correction is exact within the geometric bound.
\end{remark}

\begin{remark}[The plateau bounds the applicability]
\label{rem:plateau}
Real runs converge to some $\gamma_{\infty}>0$. Writing $\tilde\gamma_{t}=\gamma_{t}-\gamma_{\infty}$ and assuming \Bplat{}, the same argument gives $\Exp[\PI(T)]\le t_{0}+\gamma_{\infty}(T-t_{0})+\tau\tilde\gamma_{0}/\lambda$. The middle term is linear in $T$ and independent of both $\tau$ and $\lambda$. At long horizons, the plateau height therefore dominates the comparison and reflects neighbourhood quality. The transient dominates while $T-t_{0}\ll\tau\tilde\gamma_{0}/(\lambda\gamma_{\infty})$, defining the regime in which the selection-cost comparison of \Cref{prop:critical} is informative.
\end{remark}

\begin{remark}[How to read $\nstar$]
\label{rem:nstar}
$\nstar$ equates two upper bounds of different tightness and provides an order-of-magnitude estimate of the crossing. It depends on the quality deficit and encoder slope. The numerator contains two terms. $(r-1)\trep$ overtakes $rc_{0}$ once $r$ exceeds $1+c_{0}/\trep$, which on the constants of \Cref{tab:timing} is under two and a half percent. At the measured $r=1.03$, the terms are $0.044$ and $0.054$\,s and both contribute materially. Their sensitivities differ: halving the encoder slope doubles $\nstar$, whereas doubling the quality deficit $r-1$ raises it by about half because $rc_{0}$ remains. We report an interval that propagates measurement error in both quantities. Empirical consistency includes a measured crossing inside that interval, or a predicted crossing beyond the largest available instance accompanied by no observed crossing; the unknown tightness of the bounds limits the precision of either comparison.
\end{remark}

\subsection{Memory accounting}
\label{app:memory}

Memory accounting separates the network input, resident raster state and host-side index. The network reads a $128\times128\times12$ single-precision tensor occupying $786$\,kB. The rasteriser's resident state is four times that size: it holds the grid in double precision together with the per-cell sums and counts used by incremental updates. Both reside on the device and are constant in $n$. A graph encoder stores the incidence structure and its activations on the device, requiring $O(\nnz(A))$ memory. \Cref{fig:scaling}(c) reports device memory for these selectors; on this family the graph baselines exceed the $40$\,GB budget between $n=3\times10^{4}$ and $n=10^{5}$.

FOVEA also stores three $O(n)$ arrays on the host: the layout, with two 32-bit coordinates per variable; each variable's cell identifier; and an inverted index from cells to variables, with one 32-bit identifier per variable and a table of $HW$ offsets. These arrays occupy sixteen bytes per variable, or $16$\,MB at $n=10^{6}$, shown by the dotted line in \Cref{fig:scaling}(c). Per-variable channel quantities are streamed from the solver at each refresh, keeping resident host storage at sixteen bytes per variable. The representation thus places the instance-sized arrays in host memory while keeping the network input independent of $n$.

\section{Method details}
\label{app:method}

\subsection{The budget regime}
\label{app:regime}

Holding $k$ fixed in absolute terms defines the scaling regime. Under the locality conditions of \Cref{sec:analysis}, the geometric coupling-advantage bound scales as $\sqrt{sk}/\kappa$ and remains constant in $n$ on families with constant density parameters; \Cref{prop:expander} characterises a contrasting regime. When $n\gg ks$, the destroy fraction $\rho=k/n$ lies below the predicted peak location $\rhostar=1/s$ of the coupling-advantage envelope. These instances occupy the branch that rises with $\rho$, where the geometric factor is $\sqrt{sk}/\kappa$. We evaluate the empirical trend motivated by \Cref{cor:rhostar} at $n\le10^{4}$. By \Cref{cor:gstar}, a small $\rho$ also raises the saturation resolution $\gstar$, motivating the multi-scale extension in \Cref{app:foveation}. Every figure with $n$ on the horizontal axis states that $k$ is fixed, so the scaling curves describe a common absolute repair budget.

\subsection{Entity layers}

\Cref{tab:entities} instantiates the layer of \Cref{sec:entity} for the four families. Two placements distribute mass across cells, so \eqref{eq:entity-decode} is written with a mass threshold: a variable of a general mixed integer program occupies one column and as many rows as it has constraint blocks, and a scheduling operation occupies a horizontal run proportional to its duration.

\begin{table}[h]
\caption{The entity layer per family. \emph{Arity} is $|\ent(i)|$, the number of entities a variable refers to; the layer is vacuous exactly when the arity is one and the association is the identity.}
\label{tab:entities}
\centering
\small
\begin{tabular}{@{}lllll@{}}
\toprule
Family & Entity & Placement & Arity & Released set $S$ \\
\midrule
MILP & variable & column $=$ reordered column rank, & $1$ & variables whose mass in \\
     &          & rows $=$ its constraint blocks &     & the region reaches $\vartheta$ \\[2pt]
CVRP & customer & native coordinates, one-hot & $2$ & candidate edges with both \\
     &          &                             &     & endpoints selected \\[2pt]
JSSP & operation & (interval midpoint, machine by & $2$ & start times of the operations \\
     &           & load), spread over the duration &     & under the mask, plus the \\
     &           &                                 &     & ordering variables they induce \\[2pt]
1D-BPP & (item, bin) pair & (item rank by size, bin rank & $1$ & assignment variables \\
       &                  & by residual capacity), one-hot &  & under the mask \\
\bottomrule
\end{tabular}
\end{table}

Scheduling uses the same two-level structure as routing. A disjunctive formulation carries an ordering variable $z_{o,o'}$ for each pair of operations competing for a machine, so a variable refers to two entities and the arity is two. Releasing the ordering variables induced by the selected operations allows the sub-solver to change their machine sequence along with their start times. Operations correspond to customers in routing, and ordering variables to edges. The scheduling layout is rebuilt whenever the incumbent changes because an operation's time coordinate depends on the current solution; it is the one layout computed more than once. Packing uses artificial axes, with bins ordered by residual capacity so that adjacent positions have comparable slack. Variables without spatial meaning, such as bin-used indicators, are marked non-spatial and handled by the sub-solver.

\subsection{Channels}

\begin{table}[h]
\caption{The twelve input channels. \emph{Class} determines refresh cost: for a fixed layout, static channels are computed once; state channels are updated incrementally, and dual channels require a linear programming solve and are refreshed every $\Rlp$ rounds. HGS and CP-SAT provide signals for eight channels on routing and scheduling; an availability mask identifies the remaining relaxation-dependent channels.}
\label{tab:channels}
\centering
\small
\begin{tabular}{@{}rllll@{}}
\toprule
\# & Channel & Class & Aggregation & Semantics \\
\midrule
1  & occupancy & static & count & placement mass of the cell \\
2  & solution & state & mean & incumbent variable values \\
3  & lp\_relaxation & dual & mean & relaxation values \\
4  & integrality\_gap & dual & mean & $|\inc_{i}-\tilde x_{i}|$; the RINS signal and the default score \\
5  & reduced\_cost & dual & mean & variable reduced costs \\
6  & objective\_coeff & static & mean & coefficients of the objective \\
7  & pseudocost & dual & mean & the solver's own difficulty estimate \\
8  & tightness & state & mean & mass of tight constraints attributed to $i$ \\
9  & improvement\_history & state & mean & gain when $i$ was last released \\
10 & recency & state & mean & rounds since $i$ was last released \\
11 & type\_mask & static & mean & binary $1$, integer $1/2$, continuous $0$ \\
12 & boundary & static & mean & share of $i$'s constraints reaching far across the canvas \\
\bottomrule
\end{tabular}
\end{table}

The input tensor has twelve channels for every family. Channels $3$, $4$, $5$ and $7$ carry relaxation-dependent signals only when the sub-solver is a mixed integer programming solver. General programs and packing therefore provide signals for all twelve channels, while routing and scheduling provide eight; the remaining four use neutral values and are marked unavailable. The channel-removal ablation design uses the five groups geometry $\{1,12\}$, solution $\{2,8\}$, relaxation $\{3,4,5,7\}$, history $\{9,10\}$ and entity $\{6,11\}$.

\paragraph{Normalisation.}
The channel transforms and within-instance statistics are specified in \Cref{sec:canvas}.

\subsection{Sampling}

$\Sample$ takes the heatmap $\heat\in[0,1]^{H\times W}$, the budget $k$, an oversample factor $\beta\ge1$ and a Gumbel scale $\varepsilon$, and returns a cell set. It perturbs each cell to $\log\heat_{pq}+\varepsilon g_{pq}$ with independent $g_{pq}\sim\mathrm{Gumbel}(0,1)$, then walks the cells in decreasing perturbed order, admitting them until the admitted cells hold $\beta k$ variables or the cells with $\heat_{pq}>0$ run out. A morphological closing follows, and then the singleton components are discarded.

At unit Gumbel scale, taking cells in that order samples without replacement with probability proportional to $\heat$ \citep{kool2019stochasticbeam}; varying $\varepsilon$ adjusts the concentration of the sampling distribution. When the positively responding cells run out before $\beta k$ is reached, the round uses the available cells and a smaller budget. The morphological steps control different features: closing joins components separated by a single-cell gap, and singleton pruning removes isolated cells. Both contribute to controlling the boundary-component count $\Bpar$ in \Cref{cor:blocks}.

\subsection{Truncating a region to the budget}
\label{app:truncation}

The default truncation rule preserves cell structure within the budget. It selects whole cells in decreasing cell score, then uses the per-variable score to fill the remaining budget from the next cell. The score-only alternative takes the $k$ highest-scoring variables across the sampled region and can distribute the selection across its cells. The latency accounting assumes this alternative. The two rules tend to agree for sharply peaked heatmaps and differ for flatter ones, where the perimeter penalty has greater influence. The ablation design compares both rules with a cell-level rule that omits per-variable scoring to assess the resulting selection sets. The bound in \Cref{prop:iso} applies to full region preimages, as specified in \Cref{sec:policy}.

\subsection{Multi-scale refinement}
\label{app:foveation}

When a single cell holds thousands of variables, per-variable scoring supplies much of the selection detail. The optional refinement stage re-rasterises the bounding box of $\Psel$ at the same $H\times W$ resolution and applies $\pol$ to that crop, increasing spatial detail within the selected region. The second heatmap replaces the first inside the box, and truncation proceeds as before. Refinement adds a forward pass at fixed resolution and a rasterisation whose cost depends on the variables in the crop. \Cref{cor:gstar} describes the resolution threshold up to the constant $\cg$, which is fitted on one family for the refinement ablation.

\subsection{Training}

The primary training signal imitates a local branching oracle. At each large neighborhood search step, the oracle solves an auxiliary program to find the best assignment within a Hamming ball around the incumbent. Positive labels cover only variables changed by the sub-solve, as described in \Cref{sec:incremental}. Because oracle calls are expensive, we supplement them with hindsight relabelling of traces collected under diverse behaviour policies. Each round's realised improvement is attributed to the released cells, and the network is trained with weighted binary cross-entropy against the resulting targets.

The optional reinforcement stage maximises improvement per unit time, accounting for the repair cost of larger neighborhoods. It uses REINFORCE with a leave-one-out baseline formed from $G$ masks sampled at the same state. This unbiased baseline uses the sampled returns directly and supports the shared policy across instance scales without an additional value network.

\subsection{Budget and locality thresholds for the degree-of-freedom certificate}
\label{app:rank-substitution}

Substituting representative parameters into \Cref{prop:iso} gives the budget and locality thresholds at which the certificate's sufficient condition holds. Consider a formulation in which half the rows are equalities, as in the degree constraints of a routing model and the item-assignment rows of a packing model, with $d=5$ constraints per variable, mean support $s=10$, equality fraction $\thetaE=0.5$, density parameter $\kappa=1$ and locality constant $\cdelta=1$. Pure covering or packing programs have $\thetaE=0$; the tight-inequality variant in \Cref{rem:tight} provides a heuristic for these cases.

A random destroy set of size $k$ meets about $dk$ constraints, of which roughly $\thetaE dk=2.5k$ are equalities. This estimate exceeds $k$ throughout the budget range, leaving the rank to be measured directly at these parameters.

For a single compact block, the equality rows meeting $S$ are a $\thetaE$ share of all rows meeting it, so the count is $\thetaE$ times the sum of two terms: the cut bound of \Cref{prop:iso}, and the rows lying wholly inside. With one boundary component and area $k/n$ the perimeter is $2\sqrt{\pi k/n}$, so the geometric term is $9\kappa m\delta\per=18\sqrt{\pi}\,\cdelta\kappa\,d\sqrt{k/s}\approx31.9\,d\sqrt{k/s}$, the second term is $9\kappa\cdelta^{2}d=45$, and the rows lying wholly inside number about $dk/s$. Multiplying by $\thetaE$,
\[
\text{rows}\;\approx\;0.5\bigl(31.9\,d\sqrt{k/s}+9d+dk/s\bigr)\;=\;79.8\sqrt{k/10}+22.5+0.25k .
\]
At $k=100$, the estimated count is about $300$, above the certificate threshold. The condition $\text{rows}<k$ first holds at $k\approx1.2\times10^{3}$; at $k=100$, it also holds once $\cdelta\lesssim0.3$. The reported runs use $k\le10^{3}$, placing them in the latter regime: the layout's measured $\cdelta$ falls into that range on the families whose formulations contain equality rows. Budgets in the low thousands offer another route to the threshold and extend beyond the range evaluated here.

The certificate therefore identifies a threshold governed by budget and locality. In this example, the random row count is linear in $k$ with a coefficient above one. The compact count combines a $\sqrt{k}$ boundary term with a smaller linear interior term and eventually falls below the budget. A row count below $k$ certifies rank deficiency for the compact set. For random sets, a count above the budget leaves rank undetermined, so assessing their degrees of freedom requires measuring the rank itself.

\section{Experimental protocol in full}
\label{app:experiments}

\subsection{Instances}

\begin{table}[h]
\caption{Evaluation splits. \emph{Size} denotes customers for routing, variables for mixed integer programs and operations for scheduling. \emph{Seen} records whether the family appears in training; the two held-out families test zero-shot transfer of shared weights. \citet{huang2023cllns} publish four synthetic families at $1{,}000$ to $6{,}000$ variables. We generate all four at $10^{4}$ to align the main table with a point in the scaling sweep; this size difference is part of the protocol distinction from \Cref{tab:milp}. Routing training reaches $2000$ customers. XL, AGS and the million-scale sets test size extrapolation, while CVRPLIB X evaluates the training size range. The combinatorial auction family uses CATS \citep{leyton2000cats} with the parameters of \citet{huang2023cllns}.}
\label{tab:splits}
\centering
\small
\begin{tabular}{@{}lllrl@{}}
\toprule
Family & Split & Size & Count & Seen \\
\midrule
CVRP & CVRPLIB X \citep{uchoa2017cvrplib} & $100$--$1000$ & $100$ & \\
CVRP & CVRPLIB XL \citep{queiroga2026xl} & $10^{3}$--$10^{4}$ & $100$ & \\
CVRP & \citet{arnold2019ags} & $3\times10^{3}$--$3\times10^{4}$ & $10$ & \\
CVRP & \citet{accorsi2024filo2} & $2\times10^{4}$--$10^{6}$ & $20$ & \\
CVRP & clustered, distribution shift & $5000$ & $32$ & \\
CVRP & generated training set & $N\le2000$ & $256$ & yes \\
\midrule
MILP & SC, MIS, CA, MVC after \citet{huang2023cllns} & $10^{4}$ & $400$ & yes \\
MILP & the same at twice the size & $2\times10^{4}$ & $400$ & yes \\
MILP & scaling tiers, SC and CA & $10^{5}$, $10^{6}$ & $80$ & \\
MILP & MIPLIB 2017 \citep{gleixner2021miplib} & mixed & $240$ & \\
MILP & banded (staircase) programs & banded & $20$ & \\
MILP & expander set cover (falsification) & $10^{3}$, $10^{4}$ & $40$ & \\
\midrule
JSSP & Taillard \citep{taillard1993benchmarks} & $\le2000$ & $80$ & no \\
JSSP & Large-TA \citep{dacol2022jssp} & $\le10^{5}$ & --- & no \\
\midrule
1D-BPP & BPPLIB \citep{delorme2018bpplib} & mixed & $220$ & no \\
\bottomrule
\end{tabular}
\end{table}

The four mixed integer families use the generators of \citet{huang2023cllns}: Barabási--Albert graphs for vertex cover, after \citet{song2020generallns}; Erdős--Rényi graphs for independent set; CATS arbitrary relations for auctions \citep{leyton2000cats}; and the set covering construction of \citet{wu2021rllns}. Our implementations assemble the constraint matrix directly in sparse form for the $10^{6}$ variable tier. Generator parameters follow the source papers, with the variable count adjusted to each size tier.

The scheduling transfer set reaches $10^{5}$ operations under the selected CP-SAT sub-solver. \citet{dacol2022jssp} construct instances with up to $10^{6}$ operations and report that CP-SAT returns no feasible solution on that group within six hours. An incumbent is required to initialise our LNS pipeline, so this group falls outside the evaluation with CP-SAT. The same study reports that CP Optimizer finds solutions for six of the ten million-operation instances on one core and all ten on four cores, indicating a route to extending the evaluation with another sub-solver. Taillard's set reaches $2000$ operations; Large-TA supplies the larger instances.

\subsection{Baselines}

\begin{table}[h]
\caption{The seventeen baselines grouped by the aspect of performance they assess: solver configuration, destroy policy, encoder representation, domain performance and the contribution of learning to a spatial prior.}
\label{tab:baselines}
\centering
\small
\begin{tabular}{@{}lp{0.30\textwidth}p{0.36\textwidth}@{}}
\toprule
Layer & Methods & Comparison target \\
\midrule
Solver defaults & SCIP default, SCIP aggressive \citep{bestuzheva2021scip} & the contribution of the LNS loop \\
Handcrafted destroy & random, RINS, local branching, LB-RELAX, adaptive portfolio & learned and fixed destroy rules \\
Learned destroy & CL-LNS, IL-LNS, L2D ($k{=}5$, $k{=}10$) & canvas and graph encoders \\
Domain solvers & HGS \citep{vidal2022hgs}, LKH-3 \citep{helsgaun2017lkh3}, FILO2 \citep{accorsi2024filo2}, CP-SAT \citep{ortools} & performance relative to domain solvers \\
Spatial control & region-growing on the same canvas & \textbf{learning's contribution to the spatial prior} \\
\bottomrule
\end{tabular}
\end{table}

Learning to Delegate is evaluated at both subproblem sizes, $k=5$ and $k=10$, which define different action spaces. FILO2 is listed at both published effort levels: the long configuration uses a larger iteration budget and achieves a smaller gap (\Cref{tab:cvrp-x}). \Cref{tab:cvrp-main} includes the short configuration; the long configuration's published iteration budget exceeds the $200$\,s deadline at every tier. FOVEA and the spatial control use identical cached layouts, accessed through the same keys and checked by matching digests in the evaluation harness.

\subsection{Additional results}

\begin{table}[h]
\caption{All four synthetic families at $10^{4}$ variables and at twice that, primal integral over a $200$\,s compute-charged horizon. Best in each column in bold, ours shaded. The two families the main text reports are repeated here so that the four can be read together. The pattern on vertex cover and independent set matches the one on set cover and auctions: CL-LNS is ahead at the smaller tier and the margin is small.}
\label{tab:milp-full}
\centering
\small
% Generated table body. Edit the results data, not this file.
\begin{tabular}{@{}lcc@{\hspace{1.15em}}cc@{\hspace{1.15em}}cc@{\hspace{1.15em}}cc@{}}
\toprule
Method & \multicolumn{2}{c}{MVC} & \multicolumn{2}{c}{MIS} & \multicolumn{2}{c}{CA} & \multicolumn{2}{c}{SC} \\
\cmidrule(lr){2-3}\cmidrule(lr){4-5}\cmidrule(lr){6-7}\cmidrule(lr){8-9}
& $1\times$ & $2\times$ & $1\times$ & $2\times$ & $1\times$ & $2\times$ & $1\times$ & $2\times$ \\
\midrule
SCIP, default & 39.7 & 57.2 & 71.4 & 93.8 & 74.9 & 96.4 & 58.4 & 79.3 \\
SCIP, aggressive & 32.1 & 48.6 & 60.3 & 82.5 & 63.1 & 84.7 & 46.2 & 66.8 \\
Random LNS & 18.4 & 26.9 & 15.2 & 22.7 & 68.2 & 82.6 & 34.7 & 45.1 \\
RINS & 16.8 & 24.3 & 14.6 & 21.1 & 59.4 & 73.5 & 31.9 & 42.3 \\
Local branching & 44.5 & 63.1 & 52.8 & 74.6 & 88.3 & 104.2 & 62.8 & 81.5 \\
LB-RELAX & 14.9 & 21.7 & 13.1 & 19.4 & 52.7 & 66.9 & 28.6 & 37.4 \\
IL-LNS & 12.3 & 18.9 & 11.7 & 17.2 & 47.6 & 60.8 & 25.3 & 33.8 \\
CL-LNS & \textbf{9.8} & 15.4 & \textbf{9.4} & 14.8 & \textbf{41.2} & 54.2 & \textbf{21.4} & 28.9 \\
\midrule
Region growing & 11.6 & 16.2 & 10.9 & 15.3 & 49.8 & 61.3 & 24.1 & 30.2 \\
\rowcolor{gray!12}FOVEA & 10.7 & \textbf{14.8} & 10.2 & \textbf{13.9} & 43.5 & \textbf{51.7} & 22.0 & \textbf{27.1} \\
\bottomrule
\end{tabular}

\end{table}

\begin{table}[h]
\caption{MIPLIB 2017 under the present protocol: $300$\,s, $5$\,s per re-optimisation, SCIP 8.0.1, single threaded, three seeds. \emph{Gap $<1\%$} counts instances reaching that gap within the horizon. Among the methods rerun under this protocol, FOVEA has the lowest mean gap and primal integral and the largest number of instances with gap below $1\%$. For context, \citet{yuan2025btbslns} report a smaller gap of $1.75\%$ for BTBS-LNS on the same instance set under the solver version and thread count listed in \Cref{tab:miplib}. BTBS-LNS was not rerun here, so the two protocols provide separate reference points.}
\label{tab:miplib-ours}
\centering
\small
% Generated table body. Edit the results data, not this file.
\begin{tabular}{@{}lccc@{}}
\toprule
Method & \multicolumn{3}{c}{240 instances} \\
\cmidrule(lr){2-4}
& mean gap \% & gap $<1\%$ & $\PI$ \\
\midrule
SCIP, default & 14.2 & 71 & 194.6 \\
Random LNS & 15.1 & 68 & 201.3 \\
RINS & 13.8 & 74 & 188.4 \\
LB-RELAX & 11.4 & 82 & 171.9 \\
CL-LNS & 8.9 & 96 & 152.7 \\
\midrule
Region growing & 9.7 & 91 & 158.2 \\
\rowcolor{gray!12}FOVEA & \textbf{8.4} & \textbf{99} & \textbf{148.3} \\
\bottomrule
\end{tabular}

\end{table}

\begin{table}[h]
\caption{Constants used in \Cref{prop:critical} for the set-cover estimate in \Cref{fig:scaling}(b). The encoder slope $c_{G}d$ is fitted at the four sizes where CL-LNS runs and is also plotted in that panel. The quality ratio $r$ is obtained from compute-matched runs at $n=10^{4}$ with equal round counts. FOVEA's amortised slope is nearly five hundred times smaller than the encoder slope; substituting the constants satisfies the side condition of \Cref{prop:critical} and yields a predicted crossing.}
\label{tab:timing}
\centering
\small
% Generated table body. Edit the results data, not this file.
\begin{tabular}{@{}llll@{}}
\toprule
Symbol & Meaning & Value & Source \\
\midrule
$\tnet$ & one forward pass, one CPU core & $31$\,ms & measured \\
$\tlp/\Rlp$ & amortised relaxation refresh & $12$\,ms & measured at $\Rlp=20$ \\
$c_{\mathrm{layout}}$ & reordering, per variable & $0.9\,\mu$s & measured \\
$c_{\mathrm{raster}}$ & first rasterisation, per variable & $0.4\,\mu$s & measured \\
$\trep$ & one sub-solve & $1.8$\,s & measured \\
$c_{G}d$ & graph encoder, per variable & $5.9\times10^{-6}$\,s & fitted, \Cref{fig:scaling}(a) \\
$N$ & rounds in a run & $106$ & $\lfloor(T-t_{0})/\tau\rfloor$ \\
$c_{1}/N$ & FOVEA's own amortised slope & $1.2\times10^{-8}$\,s & derived, \Cref{rem:amortised} \\
$r$ & quality ratio $\lambda_{G}/\lambda_{F}$ & $1.03$ & compute-matched at $n=10^{4}$ \\
\bottomrule
\end{tabular}

\end{table}

\begin{figure}[h]
\centering
\includegraphics{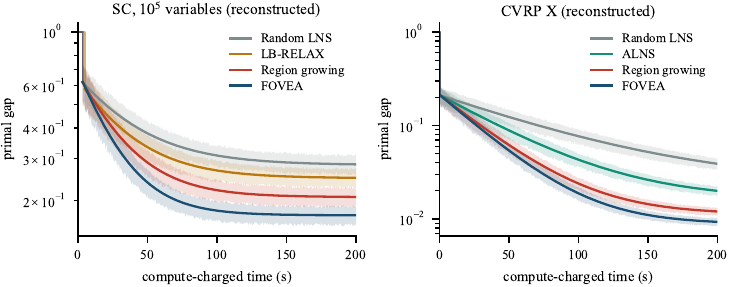}
\caption{Illustrative primal-gap curves against compute-charged time. Exponential curves are fitted so that their areas match the aggregate primal integrals in \Cref{tab:milp-main} and \Cref{tab:cvrp-main}. The shaded bands are generated illustrations, not empirical interquartile ranges over seeds. Each reconstructed curve includes an initial $p(t)=1$ segment of duration $t_{0}$, which is longer in the set-cover panel and shorter than one second in the routing panel. This segment contributes $t_{0}$ to the integral, as defined in \Cref{sec:prelim}.}
\label{fig:convergence}
\end{figure}

\begin{figure}[h]
\centering
\includegraphics{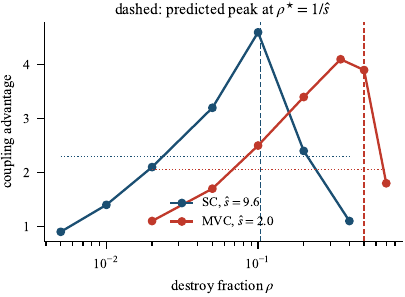}
\caption{A destroy-fraction sweep motivated by the coupling envelope in \Cref{cor:rhostar}. The empirical criterion is a peak exceeding both tails twofold, marked by the dotted lines; both plotted families meet it. On set cover, $\hat s=9.6$ places the predicted peak $\rhostar\approx0.1$ within the proved range $\rho\le1/4$. Vertex cover provides a descriptive comparison outside that range: each row has support two, so its prediction $\rhostar=0.5$ falls outside \Cref{lem:random}. Its observed peak is one grid point away, at $\rho=0.35$, within the reported seed spread. The sweep uses $n\le10^{4}$, where $\rhostar$ is reachable; at larger scales, the fixed absolute budget places runs on the ascending branch (\Cref{app:regime}).}
\label{fig:rho}
\end{figure}

\subsection{Published reference points}
\label{app:cvrp-refs}

The following tables provide published reference points under their source protocols. Their captions specify the horizons, hardware and reference values needed to interpret them alongside the matched comparisons in the main text.

\begin{table}[h]
\caption{Synthetic mixed integer programs as published by \citet{huang2023cllns}: $60$\,min horizon, $2$\,min sub-solve, SCIP 8.0.1, $100$ instances per family, mean primal integral $\pm$ standard deviation, lower is better. \Cref{tab:milp-full} gives our separate $200$\,s protocol at $10^{4}$ variables. The methods shown are those reported by the source, which provides no RINS or Gurobi results for these families. \textsc{graph} denotes its reimplementation of the learned partitioning operator of \citet{song2020generallns}. $^{\dagger}$Local branching is quoted from the source appendix and uses a $10$\,min sub-solve, compared with $2$\,min for the other rows.}
\label{tab:milp}
\centering
\small
\begin{tabular}{lcccc}
\toprule
Method & MVC-S & MIS-S & CA-S & SC-S \\
\midrule
\multicolumn{5}{l}{\emph{As published \citep{huang2023cllns}, $60$\,min, SCIP 8.0.1}}\\
Branch and bound (SCIP)   & $66.1\pm13.1$ & $222.8\pm25.9$ & $137.4\pm25.9$ & $86.7\pm37.9$ \\
Random LNS                & $38.0\pm44.8$ & $22.1\pm5.0$   & $235.6\pm34.9$ & $124.3\pm45.4$ \\
\textsc{graph}            & $42.9\pm44.0$ & $31.8\pm5.0$   & $277.7\pm36.5$ & $337.8\pm76.4$ \\
LB-RELAX                  & $57.0\pm51.2$ & $46.9\pm6.5$   & $140.5\pm18.3$ & $63.2\pm31.6$ \\
IL-LNS                    & $19.2\pm10.2$ & $19.4\pm5.8$   & $90.0\pm20.8$  & $63.2\pm34.3$ \\
RL-LNS                    & $29.6\pm11.5$ & $17.2\pm5.2$   & $249.2\pm35.9$ & $77.8\pm28.9$ \\
CL-LNS                    & $8.7\pm6.7$   & $12.8\pm5.4$   & $50.7\pm22.7$  & $26.2\pm12.8$ \\
Local branching$^{\dagger}$ & $102.2\pm35.9$ & $130.9\pm13.6$ & $272.1\pm26.9$ & $113.7\pm35.2$ \\
\bottomrule
\end{tabular}
\end{table}

\begin{table}[h]
\caption{MIPLIB 2017, $240$ instances, mean primal gap in percent, as published by \citet{yuan2025btbslns}: $300$\,s horizon, $5$\,s per re-optimisation, SCIP 7.0.3 with Gurobi 9.5.0 as the commercial reference, four cores, three seeds. The second and third columns give SCIP at longer horizons. BTBS denotes the full method; BTBS-F uses the earlier bound-tightening scheme in place of binarised tightening. Gaps use the best solution found by any compared method as the reference. The absolute gaps shown imply a relative improvement of $11.6\%$ between the two corresponding cells, compared with ten percent in the source abstract. Our evaluation of this instance set follows the separate protocol in \Cref{tab:miplib-ours}.}
\label{tab:miplib}
\centering
\small
\begin{tabular}{lccccccccc}
\toprule
& SCIP & \,$600$\,s & \,$900$\,s & U-LNS & R-LNS & FT-LNS & BTBS & BTBS-F & Gurobi \\
\midrule
As published & $15.15$ & $11.08$ & $8.79$ & $16.26$ & $15.94$ & $13.07$ & $1.75$ & $3.11$ & $1.98$ \\
\bottomrule
\end{tabular}
\end{table}

Published routing results span two size ranges and two evaluation protocols, shown separately below with their original settings. Our $200$\,s runs appear in \Cref{tab:cvrp-main}.

\begin{table}[h]
\caption{CVRPLIB X, $100$ instances with $100$ to $1000$ customers, mean percentage gap to the best known solution (lower is better). The upper block quotes \citet{vidal2022hgs}: a $2.4N$ second time limit, single threaded on an Intel Gold 6148 at $2.4$\,GHz, averages over ten runs and best known values from November 2020. That paper provides the original references for the competing heuristics; $^{\dagger}$SISR denotes its reimplementation. The lower block quotes \citet{accorsi2024filo2} on the same instances with an iteration budget, using an AMD Ryzen 5 PRO 4650GE at $3.3$\,GHz. FILO's reported gaps differ by $0.17$ points across these protocols; comparisons should be made within each block.}
\label{tab:cvrp-x}
\centering
\small
\begin{tabular}{@{}llcc@{}}
\toprule
Method & Protocol & Gap (\%) & Time (s) \\
\midrule
\multicolumn{4}{l}{\emph{As published \citep{vidal2022hgs}, $2.4N$\,s, ten runs}}\\
HGS-CVRP        & time limit, $2.4N$\,s & $0.11$ & \\
SISR$^{\dagger}$ & time limit, $2.4N$\,s & $0.19$ & \\
FILO            & time limit, $2.4N$\,s & $0.20$ & \\
HGS-2012        & time limit, $2.4N$\,s & $0.21$ & \\
KGLS            & time limit, $2.4N$\,s & $0.53$ & \\
HILS            & time limit, $2.4N$\,s & $0.66$ & \\
LKH-3           & time limit, $2.4N$\,s & $1.00$ & \\
OR-Tools        & time limit, $2.4N$\,s & $4.01$ & \\
\midrule
\multicolumn{4}{l}{\emph{As published \citep{accorsi2024filo2}, iteration budget, ten runs}}\\
FILO            & $10^{5}$ iterations & $0.37$ & $75$ \\
FILO2           & $10^{5}$ iterations & $0.34$ & $76$ \\
FILO (long)     & $10^{6}$ iterations & $0.22$ & $786$ \\
FILO2 (long)    & $10^{6}$ iterations & $0.20$ & $801$ \\
\bottomrule
\end{tabular}
\end{table}

\begin{table}[h]
\caption{Large routing instances, mean percentage gap and mean wall clock, as published by \citet{accorsi2024filo2}. Set B contains the ten instances of \citet{arnold2019ags}, with $3\times10^{3}$ to $3\times10^{4}$ customers. Set I contains the twenty instances introduced by \citet{accorsi2024filo2}, with $2\times10^{4}$ to $10^{6}$ customers. Set B gaps use CVRPLIB best known values. Set I gaps use the authors' best solutions from their experiments and therefore quantify distance to those solver-generated reference values.}
\label{tab:cvrp-xl}
\centering
\small
\begin{tabular}{@{}llcc@{}}
\toprule
Set & Method & Gap (\%) & Time (s) \\
\midrule
B, $10$ instances & FILO2 & $1.08$ & $121$ \\
B, $10$ instances & FILO2 (long) & $0.37$ & $1371$ \\
I, $20$ instances & FILO2 & $0.70$ & $370$ \\
I, $20$ instances & FILO2 (long) & $0.30$ & $3474$ \\
Lazio, $10^{6}$ customers & FILO2 & $0.40$ & $532$ \\
Lazio, $10^{6}$ customers & FILO2 (long) & $0.14$ & $3938$ \\
\bottomrule
\end{tabular}
\end{table}

Routing gaps are reported with the date of the best-known-value snapshot; the snapshots in the two tables are four years apart. Since set I uses solver-generated reference values, \Cref{sec:experiments} compares objective ratios against FILO2 at an equal deadline. The largest case is one million-customer instance, so the result at $10^{6}$ describes that individual instance.

\subsection{Empirical checks motivated by the analysis}

\begin{table}[h]
\caption{Analytical statements, associated empirical criteria and corresponding result locations. Trend and peak criteria test empirical hypotheses motivated by the bounds. The design corollaries \Cref{cor:blocks}, \Cref{cor:dimension} and \Cref{cor:gstar} determine ablation ranges in \Cref{sec:analysis}. The rank panel is illustrative; empirical evaluation of its criterion requires per-round rank and improvement logs.}
\label{tab:predictions}
\centering
\small
\begin{tabular}{@{}p{0.13\textwidth}p{0.33\textwidth}p{0.33\textwidth}p{0.10\textwidth}@{}}
\toprule
Statement & Measurement & Empirical criterion & Result \\
\midrule
\Cref{prop:iso} & locality defect $\hat\eta$ against coupling advantage, within a family & explicit cut bound holds for admissible full preimages under the stated hypotheses & \Cref{fig:rank} \\
\Cref{cor:rhostar} & advantage as a function of $\rho$ at $n\le10^{4}$, on a family with $\hat s\ge4$ & empirical peak near $\rho\approx1/\hat s$ exceeds both tails twofold & \Cref{fig:rho} \\
\Cref{prop:rank} & per-round $\rank(A_{E(S),S})/k$ against realised improvement, on the families with equality rows & all full-rank rounds have zero improvement & \Cref{fig:rank} \\
\Cref{prop:expander} & advantage over random on the expander family & coupling ratio at most $(1-2\epsilon)^{-1}$ under the hypotheses, $\epsilon<1/2$ & \Cref{tab:layout} \\
\Cref{prop:critical} & contraction rate and selection latency, giving a predicted critical size with an interval & agreement with the predicted interval, as an approximate cost-model check & \Cref{fig:scaling} \\
\bottomrule
\end{tabular}
\end{table}

\subsection{Ablations}

\begin{table}[h]
\caption{Ablation design: thirty-five studies and $117$ arms in total. \emph{Arms} counts configurations within a study, including the default. The specified seed counts are five for studies supporting main claims and three or one for the others. Numerical results supplied with the manuscript cover the five studies in \Cref{tab:ablation-results}; the remaining rows specify the broader evaluation design.}
\label{tab:ablations}
\centering
\scriptsize
\begin{tabular}{@{}llccp{0.44\textwidth}@{}}
\toprule
& Study & Arms & Seeds & Sweep \\
\midrule
\multicolumn{5}{l}{\emph{A. Layout and canvas}}\\
& A1 & $4$ & $5$ & reordering: Cuthill--McKee, spectral, random-index control, original-index control \\
& A2 & $5$ & $5$ & resolution $g\in\{32,64,128,256,512\}$ \\
& A3 & $6$ & $5$ & drop each of the five channel groups in turn, against the full stack \\
& A4 & $4$ & $3$ & cell aggregation: mean, sum, max, softmax \\
& A5 & $2$ & $3$ & twelve channels against a single score channel \\
& A6 & $2$ & $3$ & the $1\times HW$ strip with a 1-D U-Net, per \Cref{cor:dimension} \\
& A7 & $2$ & $5$ & fixed $g=128$ against $g\propto\sqrt{n}$ after \citet{gimf2025}, over four size tiers \\
& A8 & $3$ & $5$ & normalisation: per-instance, one global scale, none \\
\midrule
\multicolumn{5}{l}{\emph{B. Decoding and search}}\\
& B1 & $3$ & $5$ & two-stage decoding against cell-level only and score only \\
& B2 & $4$ & $5$ & perimeter penalty $\lambda_{\mathrm{tv}}\in\{0,0.01,0.1,1\}$ \\
& B3 & $3$ & $5$ & mask shape: free-form, connected (default), a single rectangle \\
& B4 & $4$ & $3$ & Gumbel scale $\varepsilon\in\{0,0.5,1,2\}$ \\
& B5 & $5$ & $3$ & block cap $\Bpar\in\{\text{none},1,2,4,8\}$ \\
& B6 & $3$ & $3$ & truncation score: integrality gap, uniform, relaxation value \\
& B7 & $3$ & $5$ & budget schedule: fixed, adaptive growth, the online rule of \Cref{prop:rank} \\
& B8 & $7$ & $5$ & $\rho\in\{0.01,0.02,0.05,0.1,0.2,0.5\}$ at $n\le10^{4}$, and fixed $k$ against fixed $\rho$ \\
& B9 & $3$ & $1$ & the same three action spaces under an exact oracle, isolating their cost \\
& B10 & $4$ & $3$ & dual refresh $\Rlp\in\{1,5,20,\infty\}$ \\
\midrule
\multicolumn{5}{l}{\emph{C. Architecture}}\\
& C1 & $2$ & $3$ & U-Net against an FCN at equal parameter count \\
& C2 & $4$ & $3$ & capacity in $\{0.5,1,5,20\}$M parameters \\
& C3 & $4$ & $3$ & depth in $\{2,3,4,5\}$ at $1$M parameters \\
& C4 & $2$ & $3$ & two CoordConv planes appended to the twelve channels \\
& C5 & $3$ & $5$ & GroupNorm, BatchNorm, LayerNorm \\
& C6 & $4$ & $3$ & refinement off, and on with trigger $\alpha\in\{2,4,8\}$ \\
\midrule
\multicolumn{5}{l}{\emph{D. Training}}\\
& D1 & $3$ & $5$ & imitation only, imitation then reinforcement, reinforcement only \\
& D2 & $5$ & $5$ & training set size in $\{64,128,256,512,896\}$ \\
& D3 & $3$ & $5$ & oracle: exact local branching, LB-RELAX, a greedy selector \\
& D4 & $4$ & $3$ & hindsight relabelling window $w\in\{0,1,2,4\}$ \\
& D5 & $2$ & $5$ & single-family against joint multi-family training \\
& D6 & $2$ & $3$ & curriculum on instance size \\
\midrule
\multicolumn{5}{l}{\emph{E. Transfer}}\\
& E1 & $2$ & $5$ & trained at $N\le2000$ customers, tested to $10^{6}$ \\
& E2 & $3$ & $5$ & zero-shot to the two held-out families \\
& E3 & $2$ & $5$ & uniform-trained against mixed-trained, on a clustered test set \\
& E4 & $2$ & $3$ & one shared checkpoint against a per-family fine-tune of it \\
& E5 & $3$ & $3$ & the expander family, testing the coupling-advantage limit in \Cref{prop:expander} \\
\bottomrule
\end{tabular}
\end{table}

The block-cap sweep follows \Cref{cor:blocks}. With $k=100$ and $s\approx10$ on set cover, the threshold $\Bdag\asymp k/(\cdelta^{2}s)$ places the useful range at a handful of blocks. The selected arms cover that range. The stronger condition $k\gtrsim\Bpar^{2}s$ would instead place the estimated range below one block, illustrating the effect of the threshold's dependence on the block count.

\subsection{Metrics and compute accounting}
\label{app:metrics}

Alongside the primal integral in \Cref{sec:prelim}, the protocol defines the confined primal integral $\CPI(T)=\int_{0}^{T}p(t)e^{t/\alpha}\,dt$ with $\alpha<0$ \citep{berthold2021confined}, which exponentially discounts late progress. Extending a run from $T$ to $2T$ can change the ranking induced by $\PI$. The limit of $\CPI$ as $T\to\infty$ instead provides a finite, discounted measure of the full trajectory, although finite-horizon rankings can still change before that limit. The tables use $\PI$; comparison with $\CPI$ would assess sensitivity to the deadline.

A recorded trace supports three compute conventions. \emph{Compute-charged} includes selection and repair time and is used for our tables and figures. \emph{Compute-free} counts repair time alone, isolating behavior when selection overhead is excluded. \emph{Compute-matched} fixes the round count and compares neighborhood quality separately from its cost; the quality ratio $r$ in \Cref{tab:timing} uses this convention. Comparing the conventions separates the effect of selection latency from the improvement achieved per round.

\subsection{Horizon, threading and statistics}

The evaluation criteria were fixed before the runs: a primary target of a $10\%$ primal-integral reduction over the best runnable baseline at large scales; a secondary target of staying within $10\%$ of CL-LNS at the smaller scale; and a $5\%$ reduction over region growing to assess the contribution of learning. The transfer comparison uses random selection as its reference. Limited benefit was expected on bin packing because its layout has limited metric structure.

The main experiments use a $200$\,s horizon to study the transient regime in which selection cost affects the primal integral. The MIPLIB evaluation uses the separate protocol stated in \Cref{tab:miplib-ours}. \Cref{rem:plateau} distinguishes this transient behavior from the longer-horizon regime dominated by the plateau; hour-scale performance requires evaluation at that horizon. All methods in our evaluation run single threaded, keeping repair parallelism fixed when comparing selection costs.

Configurations in the main tables use five policy seeds. The ablation design specifies the seed counts in \Cref{tab:ablations}: five for studies supporting main claims, three for the remaining studies, and one for the study requiring an exact oracle solve in every arm. For stochastic sub-solvers, each policy seed is paired with a solver seed shared across methods. Aggregates are means and standard deviations over instances. Significance is assessed by a Wilcoxon signed-rank test on paired per-instance differences, with Holm correction across the methods in a table.

\clearpage
\section{Related work}
\label{app:related}

FOVEA draws on learned destroy operators, visual representations of combinatorial state, and shared models for multiple problem families.

\subsection{Summary}
\label{app:related-summary}

The comparison follows three axes: the objects over which a destroy action is defined, the cost of encoding the search state, and the interface used to introduce a new problem family. FOVEA defines a cell mask on a fixed grid and decodes it into variables, with classical destroy signals supplied as input channels. The subsections below compare this interface with variable masks, visual operators, and shared models.

Two close visual approaches clarify the design choices: \citet{gimf2025} address scale transfer through resolution scaling, and \citet{vitsp2025} select bounding boxes on rendered tours for exact repair. We use input normalisation on a fixed canvas and compare cell masks with rectangular regions. \Cref{app:related-theory} distinguishes our single-neighborhood analysis from the convergence result for the search in \citet{burns2025lnls}.

\subsection{Learned destroy operators}

The variable-constraint incidence matrix supports several representations of a destroy decision. \citet{song2020generallns} arrange one row per integer variable with the incumbent value appended as a column, reduce it by PCA, and pass the result to a fully connected network that partitions the integer variables into $k$ disjoint subsets. They train by multiclass classification against the best of several sampled decompositions and, in a second variant, by REINFORCE. Both this approach and \Cref{sec:layout} use the constraint matrix as input: row-wise PCA and a multilayer perceptron process it as feature vectors, while FOVEA reorders the matrix to encode adjacency and applies a convolutional network.

Subsequent methods commonly use a bipartite graph encoder in the manner of \citet{gasse2019branching} and score each variable. Within this representation, local branching provides an imitation target \citep{sonnerat2021neurallns}; it is also the expert used by FOVEA (\Cref{app:method}). Its computational cost motivates the relaxation-based surrogate of \citet{huang2023lbrelax} and our use of hindsight relabelling. Other approaches use reinforcement learning from improvement rewards \citep{wu2021rllns} or a contrastive objective over positive and negative neighborhoods \citep{huang2023cllns}. The contrastive method reports the strongest results on the synthetic families of \Cref{tab:milp} under its evaluation protocol.

An extension of this representation uses a tripartite graph and a branching-guided component \citep{yuan2025btbslns}. That work reports results on synthetic integer programs, as well as the MIPLIB 2017 results quoted in \Cref{tab:miplib}. The synthetic comparisons use different instance distributions and evaluation protocols, so the reported values do not establish a direct ranking across the two studies.

Routing also supports a coarser action space. \citet{li2021l2d} use a transformer to score $O(N)$ candidate route neighborhoods and select one subproblem, rather than score individual variables.

A per-variable scoring head has $n$ outputs, while the routing selector above scores $O(N)$ candidates. The graph encoders discussed here incur $O(\nnz(A))$ cost per message-passing layer and store the graph and its activations on the device. These computation and memory costs can grow with instance size even when the destroy budget is fixed.

\citet{xu2026inch} formalise the destroy set as a binary mask $m\in\{0,1\}^{n}$ with $m_{i}=1$ when variable $i$ is released. Their coordinatewise sampling treats variables independently; the masks come from classical rules or from the logits and gradients of a repair model. Their setting uses learned repair for constraint satisfaction. FOVEA trains a separate destroy policy and uses an exact repair formulation for mixed integer programming. It represents the mask as a field over a grid, where adjacency supports connectivity, perimeter, morphological post-processing and the isoperimetric analysis of \Cref{sec:analysis}.

\subsection{Rasterised and visual combinatorial optimisation}

Visual combinatorial optimisation uses rendered state for construction, diagnosis, and iterative improvement.

\citet{zhao2021bpp} use a height map with a convolutional policy for online three-dimensional bin packing, using the natural grid geometry of the problem. The follow-up work of \citet{zhao2022pct} uses a packing configuration tree to represent placements beyond those reachable on a fixed-resolution grid. The resolution sweep in \Cref{tab:ablations} examines the corresponding representation tradeoff for FOVEA.

\citet{ling2020tsp} render travelling salesman instances as images and predict a tour with a convolutional network, and \citet{ling2021milp} do the same for 0-1 programs, solving them end to end. Both use the rendered instance to construct a solution. \citet{steever2022mip} encode constraint matrices as images for structural analysis and instance classification. \citet{vnsolver2023} render a graph under a chosen layout and classify the image with a standard vision backbone to decide whether the graph is Hamiltonian, outperforming a Graphormer baseline. VN abbreviates vision-based neural. Their comparison of circular, spiral, and random layouts identifies layout as an important design choice. Our layout ablation examines its effect on neighborhood selection.

\citet{gimf2025} study resolution scaling in a model that fuses an instance image with graph features for multi-objective routing and knapsack. They increase the total pixel count linearly with instance size, so the image side length scales as $\sqrt{n}$. Their ablation reports a slight in-distribution improvement from this scaling and a much larger one out of distribution: without it, both their full model and a simpler single-stream fusion variant fall below CNH, their graph-based baseline. The out-of-distribution result identifies a scale-transfer problem relevant to a fixed canvas. Their diagnosis attributes the decline to a shift in occupied-pixel density, including density within a patch, relative to training. FOVEA uses input normalisation to address this shift (\Cref{app:method}), with resolution scaling included as a comparison in the resolution ablation.

\citet{vitsp2025} prompt a vision-language model to propose axis-aligned bounding boxes on a rendered travelling salesman tour and repair the interior of each box exactly. This visual improvement operator is closely related to FOVEA. Their action space is a small number of rectangles, whereas a segmentation mask is an arbitrary field over cells. They use a general-purpose model through an API without task-specific training, with latency measured in seconds; FOVEA uses a small trained network with millisecond CPU latency. Their setting is one problem with native coordinates, while the canvas protocol is shared across four families. \Cref{rem:iso-role} shows that axis-aligned squares attain the grid isoperimetric optimum at a given area. Free-form masks additionally represent destroy sets with multiple components, whose coupling cost is quantified by \Cref{cor:blocks}.

\citet{visualdiffusion2025} observe that rasterisation decouples inference cost from instance size. FOVEA formalises this observation, accounts for setup and update costs in \Cref{sec:canvas}, and evaluates it at $10^{6}$ variables.

\subsection{Constructive neural solvers}

Constructive neural solvers generate an initial solution, typically by autoregressive decoding over an attention encoder \citep{kool2019attention}, with rollout baselines exploiting problem symmetry \citep{kwon2020pomo} or a heavy decoder for large-scale generalisation \citep{luo2023lehd}. FOVEA takes an incumbent as input. In the evaluated families, a cheap construction supplies this incumbent in a small fraction of the time budget, leaving most of the horizon for improvement. The operator comparisons therefore start from an existing solution. Strong domain solvers, including HGS and FILO2 for routing, appear in \Cref{tab:baselines} and run end to end on the same clock.

\subsection{Cross-problem models}

Shared models for multiple combinatorial problems use several interface designs. Heterogeneous bipartite graphs, in the lineage of \citet{gasse2019branching}, describe a new family by adding node and edge types. Adapters \citep{drakulic2024goal} keep one backbone of mixed-attention blocks and give each family its own lightweight input and output modules over a shared codebook, so a new family is a new adapter. Token sequences \citep{unico2025} recast solving as a Markov decision process and tokenise its trajectories, so a new family is a new tokenisation. Attribute vectors \citep{zhou2024mvmoe,berto2024routefinder} describe variants of one problem by switching features on and off, providing an interface for variants within routing. These approaches adapt the encoder interface to the problem family.

The canvas interface uses a fixed spatial grid whose adjacency encodes interaction. A new family supplies an entity layout and channel semantics while sharing the network parameters. The zero-shot transfer protocol in \Cref{sec:experiments} evaluates these shared parameters on held-out families.

Two terminology distinctions clarify the comparison. \citet{l2seg2025} learn to segment for vehicle routing by partitioning a tour into stable and unstable edge sequences. This is one-dimensional sequence labelling, while FOVEA predicts a field over a grid. The UniCO cited above is the work of \citet{unico2025}. A separate ICLR 2025 paper with the same name reduces problems to a matrix-encoded general TSP.

\subsection{Machinery and theoretical foundations}
\label{app:related-theory}

The policy is a U-Net \citep{ronneberger2015unet} in the fully convolutional lineage of \citet{long2015fcn}, with GroupNorm \citep{wu2018groupnorm} to normalise samples independently within batches that mix instance scales. The sampler is the Gumbel-top-$k$ construction of \citet{kool2019stochasticbeam}, with exploration controlled by the Gumbel noise scale. \Cref{app:method} specifies the sampling distribution and the role of the noise scale.

Classical destroy heuristics supply input signals, training targets, and comparison operators. Relaxation-induced neighborhood search \citep{danna2005rins} releases the variables on which the incumbent and the relaxation disagree, providing channel $4$ of \Cref{tab:channels}. Local branching \citep{fischetti2003localbranching} searches a Hamming ball and serves as our oracle. Adaptive operator portfolios \citep{pisinger2007alns} choose among fixed rules by measured performance and serve as a baseline. These roles connect the canvas policy to established neighborhood-search methods.

On the theory side, \citet{burns2025lnls} give the first non-asymptotic convergence result for a hybrid large neighborhood local search, by reduction to a block Langevin diffusion. Their analysis uses a fixed partition of a continuous relaxation and covers randomised or cyclic block selection. Its stated scope excludes discrete decision variables and strongly non-convex potentials, including mixed integer programming. The analysis of \Cref{sec:analysis} characterises the coupling and improvement conditions of a single neighborhood choice. Convergence of the resulting search trajectory remains a separate question.

\clearpage
\section{Notation}
\label{app:notation}

\begin{table}[h]
\centering
\small
\begin{tabular}{@{}ll@{\hskip 1.4em}ll@{}}
\toprule
Symbol & Meaning & Symbol & Meaning \\
\midrule
$n,m$ & variables, constraints & $\lay$ & variable layout \\
$s_{j},s$ & support size, its mean & $\ents,\ent$ & entity set, association \\
$\seff$ & truncated mean support size & $H,W,g$ & canvas height, width, resolution \\
$d$ & constraints per variable & $\Vpq$ & variables in cell $(p,q)$ \\
$k,\rho$ & destroy budget, fraction & $\canvas^{(t)}$ & the canvas at round $t$ \\
$\rhostar$ & coupling-advantage envelope & $\heat$ & predicted heatmap \\
            & peak location & & \\
$\region,\SR$ & admissible region, its preimage & $\Psel$ & selected cells \\
$S$ & destroy set, $S\subseteq\SR$ & $\regfam$ & admissible region family \\
$\Srand$ & uniform random $k$-subset & $\pol$ & the policy \\
$\inc$ & incumbent & $\Rlp$ & dual refresh period \\
$\sub{S}$ & restricted subproblem & $\varepsilon$ & Gumbel scale \\
$\gain(S)$ & best improvement from $S$ & $\vartheta$ & entity mass threshold \\
$\cut{S}$ & straddling constraints & $\tsel,\trep$ & selection, repair time \\
$\cutE{S},\inE{S}$ & equality rows cut, contained & $\tau$ & round time $\tsel+\trep$ \\
$\thetaE$ & equality fraction & $t_{0}$ & time to first feasible \\
\midrule
$\delta,\eta$ & locality scale and defect \Aloc{} & $\PI(T)$ & primal integral \\
$\kappa$ & density constant \Aden{} & $\CPI(T)$ & confined primal integral \\
$C,B$ & isoperimetric constant and & $\gamma_{t},\gamma_{\infty}$ & primal gap, plateau \\
      & boundary-component cap, \Areg{} & $L$ & Lipschitz constant \Alip{} \\
$\cdelta$ & locality constant, $\delta=\cdelta\sqrt{s/n}$ & $\lambda$ & gap contraction rate \\
$\cg$ & resolution constant of $\gstar$ & $\nstar$ & critical instance size \\
$\Bpar$ & boundary components & $\gstar$ & saturation resolution \\
$\Bcc$ & connected components & $\rhomax$ & maximum variable density \\
$\Bdag,\Bddag$ & block thresholds & $\Lampar$ & mean density on the boundary \\
$\per(\region)$ & $\mathcal{H}^{1}(\partial\region)$ & $\Lammax$ & global maximum density \\
$\epsilon$ & expansion defect, \Cref{prop:expander} & $\zeta$ & budget efficiency \\
$\Bstar$ & boundary components of the & $r$ & quality ratio $\lambda_{G}/\lambda_{F}$ \\
         & minimiser in \Cref{def:per-of-set} & $c_{0},c_{1}$ & fixed and amortised per-round cost \\
$g_{0}$ & reference resolution, $g_{0}\gg g$ & $N$ & rounds in a run \\
$\beta$ & oversample factor & $\alpha$ & confinement rate of $\CPI$ \\
\bottomrule
\end{tabular}
\end{table}

The locality constant $\cdelta$ and the resolution constant $\cg$ use distinct subscripts on $c$. The boundary component count $\Bpar$, constrained by \Areg{}, counts each boundary curve, while the connected component count $\Bcc$ counts connected regions: a region with one hole has one component and two boundaries. The mean density along a mask boundary $\Lampar$ and the global maximum $\Lammax$ also serve different roles. Substituting the latter into \Cref{cor:gstar} overestimates the saturation resolution by $\sqrt{\Lammax/\Lampar}$ on clustered instances where the uniformity assumption fails. The contraction rate in \Cref{prop:critical} is $\lambda$, the destroy fraction is $\rho$, the dual refresh period is $\Rlp$, and the region is $\region$. The predicted destroy fraction at the peak of the coupling-advantage envelope is $\rhostar$. The symbol $\epsilon$ denotes the expansion defect in \Cref{prop:expander} and the relative error tolerance in \Cref{cor:gstar}. The entity mass threshold is $\vartheta$, distinct from the policy parameter $\theta$ in $\pol$.

\end{document}